\documentclass[11pt]{article}

\usepackage{amsfonts}
\usepackage{graphicx}
\usepackage{amsmath,amssymb,amsthm,framed}
\usepackage{epsf,float}
\usepackage{epsfig}
\usepackage{fullpage}
\usepackage{graphics}
\usepackage{color}
\usepackage[normalem]{ulem}
\usepackage{blkarray}
\usepackage{booktabs}
\usepackage{bbm}
\usepackage[linesnumbered,ruled,vlined]{algorithm2e}
\usepackage{enumerate}
\usepackage{cite}

\newcommand {\apx} {{\sf APX}}
\newcommand {\np} {{\sf NP}}

\DeclareMathOperator*{\supp}{supp}
\DeclareMathOperator*{\rd}{{Red}}

\newtheorem{theorem}{Theorem}[section]

\newtheorem{lemma}[theorem]{Lemma}
\newtheorem{observation}[theorem]{Observation}
\newtheorem{corollary}[theorem]{Corollary}
\newtheorem{proposition}[theorem]{Proposition}

\newtheorem{definition}{Definition}[section]
\newtheorem{example}[theorem]{Example}

\newtheorem{remark}[theorem]{Remark}
\newtheorem*{theorem1}{Theorem~\ref{thm:beta=gamma} (restated)}
\newtheorem*{theorem3}{Theorem~\ref{thm:apxHard} (restated)}

\newcommand{\Zplus}{{{\mathbb Z}_+}}

\newcommand{\smallPi}{f}
\newcommand{\bigPi}{\Pi}

\begin{document}
\title{Perfect phylogenies via branchings in acyclic digraphs\\ and a generalization of Dilworth's theorem}

\author{Ademir Hujdurovi\' c$^{a,b}$\thanks{E-mail: \texttt{ademir.hujdurovic@upr.si}}
\and
Edin Husi\' c$^{c,b}$\thanks{E-mail: \texttt{e.husic@lse.ac.uk}}
\and
Martin Milani\v c$^{a,b}$\thanks{E-mail: \texttt{martin.milanic@upr.si}}
\and
Romeo Rizzi$^{d}$\thanks{E-mail: \texttt{romeo.rizzi@univr.it}}
\and
Alexandru I.~Tomescu$^{e}$\thanks{E-mail: \texttt{tomescu@cs.helsinki.fi}}
}

\maketitle

\vspace{-0.5cm}
\begin{center}
$^a$ University of Primorska, UP IAM, Koper, Slovenia\\

$^b$ University of Primorska, UP FAMNIT, Koper, Slovenia \\

$^c$ London School of Economics, Department of Mathematics, London, United Kingdom\\

$^d$ University of Verona, Department of Computer Science,
Verona, Italy\\

$^e$ Helsinki Institute for Information Technology HIIT, Department of Computer Science, University of Helsinki, Finland
\end{center}

\begin{abstract}
Motivated by applications in cancer genomics and following the work of Hajirasouliha and Raphael (WABI 2014), Hujdurovi\'c et al.~(IEEE TCBB, to appear) introduced the minimum conflict-free row split (MCRS) problem: split each row of a given binary matrix into a
bitwise OR of a set of rows so that the resulting matrix corresponds to a perfect phylogeny and has the minimum possible number of rows among all matrices with this property. Hajirasouliha and Raphael also proposed the study of a similar problem, in which the task is to minimize the number of distinct rows of the resulting matrix. Hujdurovi\'c et al.~proved that both problems are NP-hard, gave a related characterization of transitively orientable graphs, and proposed a polynomial-time heuristic algorithm for the MCRS problem based on coloring cocomparability graphs.

We give new, more transparent formulations of the two problems, showing that the problems are equivalent to
two optimization problems on branchings  in a derived directed acyclic graph. Building on these formulations, we obtain new results on the two problems, including: (i) a strengthening of the heuristic by Hujdurovi\'c et al.~via a new min-max result in digraphs generalizing Dilworth's theorem, which may be of independent interest, (ii) APX-hardness results for both problems, (iii) approximation algorithms, and (iv) exponential-time algorithms solving the two problems to optimality faster than the na\"ive brute-force approach. Our work relates to several well studied notions in combinatorial optimization: chain partitions in partially ordered sets, laminar hypergraphs, and (classical and weighted) colorings of graphs.
\end{abstract}


\maketitle

\section{Introduction}

\begin{sloppypar}
A \emph{perfect phylogeny} is a rooted tree representing the evolutionary history of a set of $m$ objects.
The objects bijectively label the leaves of the tree and there are $n$ binary variables called \emph{characters},
each labeling exactly one edge of the tree. For each leaf, the set of characters that appear on the
unique root-to-leaf path is the set of characters taking value $1$ at the object labeling the leaf. While
every perfect phylogeny naturally corresponds to an $m\times n$ binary matrix having objects as rows and
characters as columns, the \emph{perfect phylogeny problem} asks the opposite question: Does a given binary
matrix correspond to a perfect phylogeny? The perfect phylogeny problem and various generalizations of it
have been extensively studied in computational biology. In this paper we study two combinatorial optimization problems,
 both generalizations of the perfect phylogeny problem, first considered by Hajirasouliha and Raphael~\cite{wabi14} and
 motivated by applications in cancer genomics.

Following the work~\cite{wabi14}, Hujdurovi\'c et al.~\cite{tcbb16} introduced the \emph{minimum conflict-free row split}
problem, which can be informally stated as follows: given a binary matrix $M$, split each row of $M$ into a bitwise OR\footnote{Here,
OR denotes the usual binary OR function, assuming that value true is represented by $1$ and value false by $0$, that is,
\hbox{$x \text{\,OR\,} y = 1$} if and only if at least one of $x$ and $y$ has value $1$.}
of a set of rows so that the resulting matrix corresponds to a perfect phylogeny and has the minimum number of rows
among all matrices with this property. To state the problem formally, we need two definitions.
\end{sloppypar}

\begin{definition}\label{def:conflict}
Given a matrix $M$, three distinct rows $r$, $r'$, $r''$ of $M$ and two distinct columns $i$ and $j$ of $M$, we denote by $M[(r,r',r''),(i,j)]$ the $3\times 2$ submatrix of $M$ formed by rows $r$, $r'$, $r''$ and columns $i$, $j$ (in this order).
Two columns $i$ and $j$ of a binary matrix $M$ are said to be {\normalfont in conflict} if there exist rows $r, r', r''$ of $M$ such that
$$M[(r,r',r''),(i,j)] = \left(
                          \begin{array}{cc}
                            1 & 1 \\
                            1 & 0 \\
                            0 & 1 \\
                          \end{array}
                        \right)\,.$$
We say that a binary matrix $M$ is{ \normalfont conflict-free} if no two columns of $M$ are in conflict.
\end{definition}

\begin{definition}
Let $M \in \{0,1\}^{m\times n}$. Label the rows of $M$ as $r_1,r_2, \dots , r_m$. A binary matrix $M' \in \{0,1\}^{m' \times n}$ is a {\normalfont row split} of $M$ if there exists a partition of the set of rows of $M'$ into $m$ sets $R_1, R_2, \dots R_m$ such that for all $i \in \{1, \dots ,m\}$, $r_i$ is the bitwise $OR$ of the binary vectors in $R_i$. The set $R_i$ of rows of $M'$ is said to be the set of {\normalfont split rows} of row $r_i$ (with respect to $M'$).
\end{definition}

For simplicity, we defined a row split as a binary matrix $M'$ for which a suitable partition of rows exists. However, throughout the paper we will make a slight technical abuse of this terminology by considering any row split $M'$ of $M$ as already equipped with an arbitrary (but fixed) partition of its rows $R_1,\ldots, R_m$ satisfying the above condition. For an example of these notions, see Fig.~\ref{fig:ex}. For the sake of clarity, from now on we omit displaying the zeros in binary matrices.

\begin{figure}[h!]
\begin{center}
\includegraphics[width=0.8\textwidth]{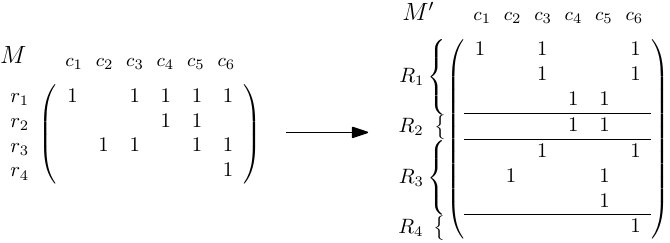}
\end{center}
\caption{An example of a binary matrix $M$ and a conflict-free row split $M'$ of $M$.}
\label{fig:ex}
\end{figure}

We denote by $\gamma(M)$ the minimum number of rows in a conflict-free row split $M'$ of $M$.
Formally, the minimum conflict-free row split problem is defined as follows:

\medskip

\begin{center}
\fbox{\parbox{0.95\linewidth}{\noindent
{\sc MinimumConflict-FreeRowSplit (MCRS):}\\[.8ex]
\begin{tabular*}{.93\textwidth}{rl}
{\em Input:} & A binary matrix $M$.\\
{\em Task:} & Compute $\gamma(M)$.
\end{tabular*}
}}
\end{center}

\medskip

We will also consider a variant of the problem, proposed by Hajirasouliha and Raphael~\cite{wabi14}, in which the task is to compute a row split $M'$ of $M$ such that the number of {\em distinct} rows in $M'$ is minimized. Let $\eta(M)$ denote the minimum number of distinct rows in a conflict-free row split $M'$ of $M$. Similarly as above, we consider the corresponding optimization problem.

\medskip

\begin{center}
\fbox{\parbox{0.95\linewidth}{\noindent
{\sc MinimumDistinctConflict-FreeRowSplit (MDCRS):}\\[.8ex]
\begin{tabular*}{.93\textwidth}{rl}
{\em Input:} & A binary matrix $M$.\\
{\em Task:} & Compute $\eta(M)$.
\end{tabular*}
}}
\end{center}

\medskip
The connection between conflict-free matrices and perfect phylogenies is well known: the rows of a binary matrix $M$ are the leaves of a perfect phylogeny if and only if $M$ is conflict-free (see~\cite{Estabrook:1975aa,gusfieldbook}). Moreover, if this is the case, then the corresponding tree can be retrieved from $M$ in time linear in the size of $M$~\cite{Gusfield91}. The intuition behind the fact that a conflict-free matrix corresponds to a perfect phylogeny is that one can map each row to a leaf of a tree, and each column to an edge, so that each row has a~$1$ exactly on those columns that are mapped to the edges on the path from the root to the leaf corresponding to the row. The forbidden $3 \times 2$ matrix from Definition~\ref{def:conflict} as a submatrix leads a contradiction, since then the two distinct edges $e_i$ and $e_j$ to which columns $i$ and $j$ are mapped, respectively, are such that $e_i$ appears both before, and after, $e_j$ on a root-to-leaf path. We refer to~\cite{wabi14,tcbb16} and to references therein for further details on the biological aspects of the MCRS and the MDCRS problems.

Another well studied family of combinatorial objects closely related to the MCRS and the MDCRS problems
are laminar set families. A set family $\mathcal{F}$ is said to be \emph{laminar} if every two sets $A,B\in \mathcal{F}$ satisfy $A \cap B = \emptyset$, $A\subseteq B$, or $B \subseteq A$. The connection with laminar families follows from the fact that a binary matrix $M$ is conflict-free if and only if the sets of rows indicating the positions of ones in the columns of $M$ form a laminar family. This connection will be exploited in Section~\ref{sec:2-approx}. Laminar families of sets play an important role in network design problems~\cite{MR1805713}, in the study of packing and covering problems~\cite{MR1627624,MR1729148,MR2436998}, and in several other areas of combinatorial optimization, see, e.g.,~\cite{schrijver}.

\medskip
In~\cite{tcbb16}, the MCRS and the MDCRS problems were proved \np-hard, a related characterization of transitively
orientable graphs was given, and a polynomial-time heuristic algorithm was proposed for the MCRS problem based on coloring cocomparability graphs (that is, complements of transitively orientable graphs). Following~\cite{tcbb16}, the main aim of this paper is to further advance the understanding of structural and computational aspects of the MCRS and the MDCRS problems.
\medskip

\begin{sloppypar}
\noindent{\bf Our results and techniques.}
The first and main result of this paper is a result showing that the MCRS and the MDCRS
problems can be equivalently formulated as two optimization problems on branchings in
a directed acyclic graph derived from the given binary matrix, the so-called {\it containment digraph}. (Precise definitions of these notions and the corresponding problems will be given in Section~\ref{sec:formulation}.) These equivalencies lead to more transparent formulations of the two problems. We will ascertain the applicability and usefulness of
these novel formulations by deriving the following results and insights about the MCRS and the MDCRS problems:
\begin{itemize}
  \item We prove a new min-max result on digraphs strengthening Dilworth's theorem on chain partitions and antichains in partially ordered sets.
      This result is described in Section~\ref{sec:min-max}, which can be read independently of the rest of the paper. This result, besides being interesting on its own as a generalization of a classical min-max result, connects well to the MCRS problem via the problem's branching formulation. The constructive, algorithmic proof of the result shows that a related problem is polynomially solvable: a problem in which only a subset of all branchings of the containment digraph is examined, namely the so-called linear branchings (branchings corresponding to chain partitions of the partial order underlying the containment digraph).
      This approach leads to a new heuristic for the MCRS problem, improving on a previous heuristic~from~\cite{tcbb16}.
  \item We strengthen the \np-hardness results for the two problems to \apx-hardness results.
 \item We complement the inapproximability results with three approximation algorithms: a~\hbox{$2$-approximation} algorithm for the MDCRS problem (implying that the problem is \apx-complete) and two approximation algorithms for the MCRS problem, the approximation ratios of which are expressed in terms of two parameters of the containment digraph, corresponding to the height and the width of the underlying partial order, respectively.
\item The branching formulations allow for the development of faster exact exponential-time solutions for the two problems,
when compared to a direct brute-force approach that follows directly from the problems' definitions.
\end{itemize}
\end{sloppypar}
\medskip

\begin{sloppypar}
\noindent{\bf Comparison with related work.}
In~\cite{wabi14}, Hajirasouliha and Raphael introduced the so-called Minimum-Split-Row
problem, in which only a given subset of rows of the input matrix needs to be split and,
roughly speaking, the task is to minimize the number of additional rows in the resulting
conflict-free row split. All results from~\cite{wabi14} actually deal with the variant of the
problem in which all rows need to be split (some perhaps trivially by setting $R_i = \{r_i\}$);
in this case, the optimal value of the Minimum-Split-Row problem coincides with
the difference \hbox{$\gamma(M)-r(M)$}, where
$r(M)$ is the number of rows of $M$. In the same paper, a lower bound on
the value of $\gamma(M)$ was derived and, in the concluding remarks of the paper,
a study of the MDCRS problem was suggested.
In subsequent works by Hujdurovi\'c et al.~\cite{tcbb16}, the MCRS
problem was introduced and several claims from~\cite{wabi14} were proved incorrect,
including an \np-hardness proof of the Minimum-Split-Row problem (which would imply \np-hardness
of the MCRS problem). However, it was shown in~\cite{tcbb16} that the
MCRS problem is indeed \np-hard, as is the MDCRS problem. Moreover, a polynomially
solvable case of the MCRS problem was identified and an efficient heuristic algorithm for
the problem on general instances was proposed, based on coloring cocomparability graphs.
\end{sloppypar}

The results of this paper improve on the previously known results about the
two problems: \np-hardness results are strengthened to \apx-hardness results,
approximation algorithms for the two problems are proposed, and the heuristic
algorithm for the MCRS problem given by Hujdurovi\'c et
al.~from~\cite{tcbb16} is improved. The key tools leading to most of these
results are the newly proposed branching formulations and the new min-max
theorem strengthening Dilworth's theorem. The min-max theorem
has a constructive algorithmic proof, leading to a polynomial-time algorithm
to compute a chain partition of a given partially ordered set equipped with a
monotone weight function such that the of sum of the maximum weights in the chains
is minimized. This result contrasts with known results
in the literature implying that two natural variants of the problem are
\np-hard: (i) the variant in which the chains used in the partition have to
be of bounded size~\cite{MR1380086,MR2386521}, and (ii) the variant in which
the weight function is not necessarily monotone, which corresponds to a
variant of the graph coloring problem known as Weighted Coloring (see,
e.g.,~\cite{MR1439871,MR2195353,MR2499496,MR3286675}), in the class of
cocomparability graphs. We refer to the remarks following
Corollary~\ref{cor:min-price-chain-partition-algorithm} in
Section~\ref{sec:min-max} for more details.
See also Figure~\ref{fig:problems} in Section~\ref{sec:conclusion}, where we
summarize the relations between the problems introduced in this paper
and several problems studied in the literature, along with the corresponding complexity
results.

\medskip

\noindent{\bf Structure of the paper.} The branching formulations of the two
problems are given in Section~\ref{sec:formulation}. A strengthening of Dilworth's
theorem and its connection to the MCRS problem is discussed in Section~\ref{sec:Dilworth}.
\apx-hardness proofs and approximation algorithms are presented in Section~\ref{sec:approx}.
We conclude the paper with a summary and some questions for future research
in Section~\ref{sec:conclusion}.

\medskip

\noindent{\bf Remark on notation.} A {\em binary matrix} $M\in \{0,1\}^{m\times n}$ is a matrix having $m$ rows and $n$ columns, and all entries $0$ or $1$. Each row of such a matrix is a vector in $\{0,1\}^n$; each column is a vector in $\{0,1\}^m$.
We will usually denote by $R_M = \{r_1,\ldots, r_m\}$ and $C_M = \{c_1,\ldots, c_n\}$ the (multi)sets of rows and columns of $M$, respectively.
The entry of $M$ at row $r_i$ and column $c_j$ will be denoted by $M_{i,j}$ or $M_{r_i,j}$ when appropriate.
For brevity, we will often write ``the number of distinct rows (resp., columns) of $M$'' to mean
``the maximum number of pairwise distinct rows (resp., columns) of $M$''. Two rows (resp., columns) are considered distinct
if they differ as binary vectors. All binary matrices in this paper will be assumed to contain no row whose all entries are~$0$.

In our proofs and constructions we will often simplify the binary matrix $M$ under consideration by working instead with the matrix denoted by $\rd(M)$, obtained by taking from $M$ exactly one copy from each set of identical columns.

\smallskip
An extended abstract of this work appeared in the proceedings of WG 2017~\cite{WG2017}.

\section{Formulations in Terms of Branchings in Directed Acyclic Graphs}\label{sec:formulation}

In this section, we are going to formulate the MCRS and the
MDCRS problems in terms of branchings in directed acyclic graphs (DAGs).
First, we give the necessary definitions.

\begin{definition}
Let $D = (V,A)$ be a DAG. A {\em branching} of $D$ is a subset $B$ of $A$ such that $(V,B)$ is a digraph in which for each vertex $v$ there is at most one arc leaving $v$.
\end{definition}

The following construction (see, e.g.,~\cite{wabi14}) can be performed on any given binary matrix $M$ and results in a directed acyclic graph. Given a column $c_j\in C_M$, the {\em support of $c_j$} is the set defined as \hbox{$\{r_i\in R_M: M_{i,j} = 1\}$} and denoted by $\supp_M(c_j)$. Given a binary matrix $M\in \{0,1\}^{m\times n}$, the {\em containment digraph} $D_M$ of $M$ is the directed acyclic graph with vertex set $V=\{\supp_M(c): c\in C_M\}$
and arc set $A = \{(v,v'): v,v'\in V \land v\subset v'\}$ where $\subset$ is the relation of proper inclusion of sets.
See Fig.~\ref{fig:DM} for an example.

\begin{figure}[h!]
\begin{center}
\includegraphics[width=\textwidth]{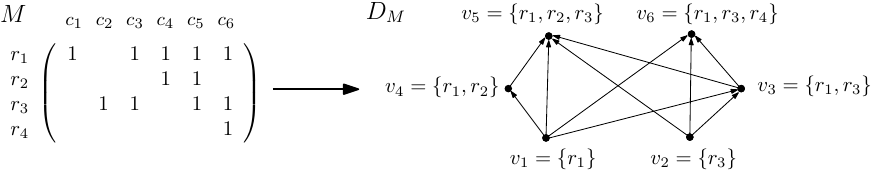}
\end{center}
\caption{An example of a binary matrix $M$ and its containment digraph $D_M$.}
\label{fig:DM}
\end{figure}

Let $M\in\{0,1\}^{m\times n}$ be a binary matrix, let $D_M=(V,A)$ be the containment digraph of $M$, and let $B$ be a branching of $D_M$.
For a vertex $v\in  V$, we denote by $N^-_B(v)$ the set of all vertices
$v'\in V$ such that $(v',v)\in B$. A {\em source} of $B$ is a vertex not entered by any arc of $B$.
For a vertex $v \in V$, an element $r \in v$ (that is, a row of $M$) is said to be {\em covered} in $v$ with respect to $B$
(or just {\it $B$-covered}) if $r \in \cup \, N^-_B(v)$. (When it is clear to which branching we are referring to, we will say just that ``$r$ is covered in $v$''.) Analogously, we say that $r\in v$ is {\em uncovered} in $v$ with respect to $B$ if $r$ is not covered in $v$. A {\em $B$-uncovered pair} is a pair $(r,v)$ such that $r$ is a row of $M$, $v$ is a vertex of $D_M$ (that is, the support of a column of $M$), $r \in v$, and $r$ is uncovered in $v$ with respect to $B$. For a row $r$ of $M$, we will denote by $U_B(r)$ the set of all $B$-uncovered pairs with first coordinate $r$, and by $U(B)$ the set of all $B$-uncovered pairs. To illustrate these notions, we elaborate further on the example from Fig.~\ref{fig:DM} in Fig.~\ref{fig:example}, where two branchings $B_1$ and $B_2$ of the arc set of $D_M$ are depicted, together with uncovered pairs $(r,v)$ with respect to each of the two branchings.

\begin{figure}[h!]
\begin{center}
\includegraphics[width=\textwidth]{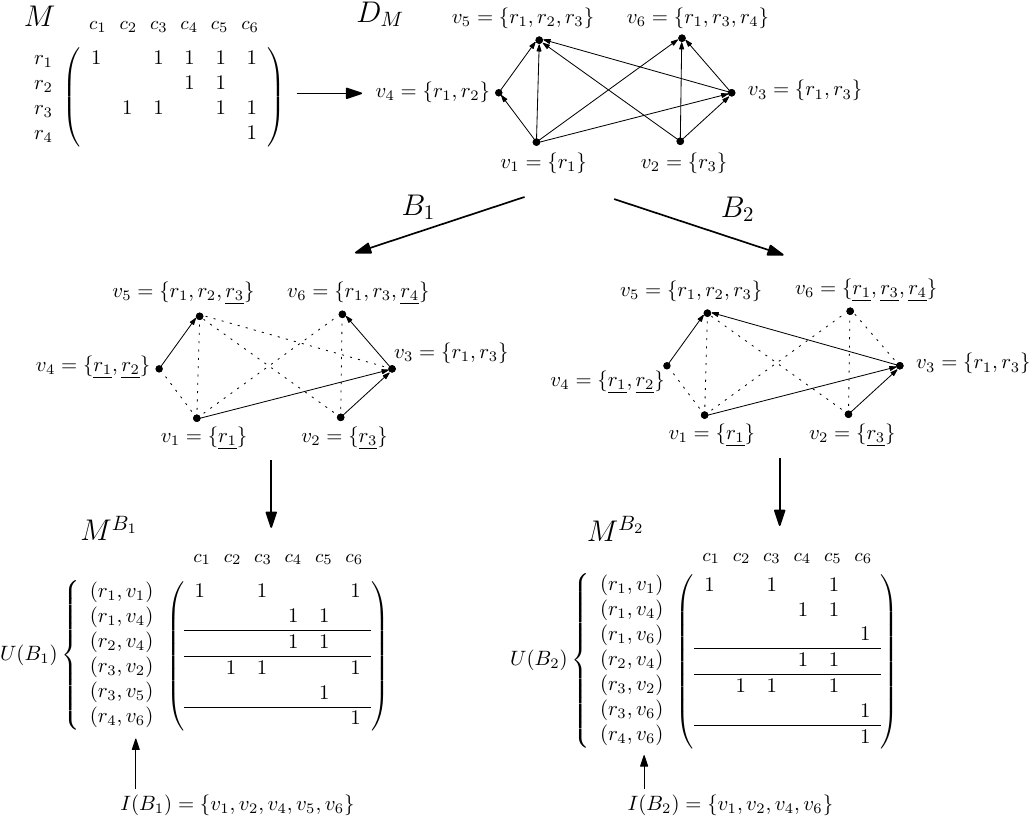}
\end{center}
\caption{An example of a binary matrix $M$, its containment digraph $D_M$ with two branchings $B_1$ and $B_2$, and the resulting row splits of $M$ (the $B_1$-split and $B_2$-split). The row split $M^{B_1}$ is an optimal solution to the MCRS problem given $M$. The row split $M^{B_2}$ is an optimal solution to the MDCRS problem given $M$.
Pairs $(r,v)$ for which $r$ is underlined as an element of $v$ in the figure showing $B_1$ (respectively, $B_2$)
are exactly the uncovered elements with respect to $B_1$ (respectively, $B_2$).}
\label{fig:example}
\end{figure}

For a branching $B\subseteq A$, we say that a vertex $v\in V$ is {\em $B$-irreducible} if there exists some element $r\in v$ that is uncovered in $v$ with respect to $B$ (equivalently, if $v \not\in \cup \, N^-_B(v)$).
In particular, every source of $B$ is $B$-irreducible. We denote by $I(B)$ the set of all $B$-irreducible vertices; see Fig.~\ref{fig:example} for an example.

We denote by $\beta(M)$ the minimum number of elements in $U(B)$ over all branchings $B$ of $D_M$. Similarly, we denote with $\zeta(M)$ the minimum number of elements in $I(B)$ over all branchings $B$ of $D_M$. The corresponding optimization problems are the following:

\medskip
\noindent
\begin{minipage}[t]{0.5\textwidth}
\begin{center}
\fbox{\parbox{0.95\linewidth}{\noindent
{\sc MinimumUncoveringBranching (MUB):}\\[.8ex]
\begin{tabular*}{.93\textwidth}{rl}
{\em Input:} & A binary matrix $M$.\\
{\em Task:} & Compute $\beta(M)$.
\end{tabular*}
}}
\end{center}
\end{minipage}
\hfill
\begin{minipage}[t]{0.5\textwidth}
\begin{center}
\fbox{\parbox{0.95\linewidth}{\noindent
{\sc MinimumIrreducingBranching (MIB):}\\[.8ex]
\begin{tabular*}{.93\textwidth}{rl}
{\em Input:} & A binary matrix $M$.\\
{\em Task:} & Compute $\zeta(M)$.
\end{tabular*}
}}
\end{center}
\end{minipage}

\medskip
The announced equivalence between the MCRS and the MUB problems, and between the MDCRS and the MIB problems is captured in the following theorem.
We denote by $\omega$ any real number such that there exists an ${\mathcal{O}}(n^{\omega})$ algorithm for multiplying two $n\times n$ binary matrices (e.g., $\omega=2.3728639$~\cite{MR3239939}).

\begin{sloppypar}
\begin{theorem}\label{thm:beta=gamma}
For every binary matrix $M\in \{0,1\}^{m\times n}$ with exactly $k$ distinct columns, the following holds:
\begin{enumerate}
  \item Any branching $B$ of $D_M$ can be transformed in time ${\mathcal{O}}(kmn)$ to a conflict-free row split of $M$ with exactly $|U(B)|$ rows and with exactly $|I(B)|$ distinct rows.

  \item Any conflict-free row split $M'\in \{0,1\}^{m'\times n}$ of $M$ can be transformed in time \hbox{${\mathcal{O}}(mn+m'k^2+k^{\omega})$} to a branching $B$ of $D_M$ such that $|U(B)|$ is at most the number of rows of $M'$ and
$|I(B)|$ is at most the number of distinct rows of $M'$.
\end{enumerate}
Consequently, for every binary matrix $M$, we have $\gamma(M)=\beta(M)$ and $\eta(M)=\zeta(M)$.
\end{theorem}
\end{sloppypar}

\begin{sloppypar}
Results presented in Sections~\ref{sec:heuristic},~\ref{subsec:hard}, and~\ref{subsect:height and width} will rely on Theorem~\ref{thm:beta=gamma}. Before giving a proof of the theorem, let us discuss one further consequence of it.
The theorem allows for the development of faster exact exponential-time solutions for the two problems, when compared to a direct brute-force approach that follows directly from the problems' definitions. Consider the simple approach of enumerating all possible branchings of $D_M$ and selecting the best one. Denoting by $W$ the set of vertices $u$ of $D_M$ of out-degree $d^+(u)$ at least one and disregarding polynomial factors, the time complexity of this approach is of the order ${\mathcal{O}}(\prod_{u\in W}d^+(u)) = {\mathcal{O}}(n^n) = {\mathcal{O}}(2^{n\log n})$,
where $n$ is the number of distinct columns of the input matrix $M$.
On the other hand, the time complexity of the straightforward approach to the two problems based on generating all possible row splits of $M$
cannot even be expressed as a function of $n$ only. A row with $k$ ones has at least as many splits as the number of partitions of a $k$-element set, which is the quantity counted by the Bell number $B_k$ and clearly bounded from below by $2^k$. Thus, for a matrix with $m$ rows, each with at least $n/2$ ones, the total number of row splits of $M$ is at least $2^{mn/2}$.
\end{sloppypar}

\smallskip
Theorem~\ref{thm:beta=gamma} will be proved in two steps. First, we show how to split the input matrix $M$ in a conflict-free way, given a branching $B$ of its containment digraph; the number of rows (resp., distinct rows) of the resulting row split equals the number of $B$-uncovered pairs (resp., $B$-irreducible vertices). Second, we show that any conflict-free row split $M'$ of $M$ can be reduced, by possibly deleting some rows, into a row split of $M$ obtained from some branching of $D_M$ (as in the first step).

The proof of the first part of Theorem~\ref{thm:beta=gamma} relies on the notion of a {\it $B$-split}, defined as follows.

\begin{definition}\label{def:B-split}
Let $M$ be a binary matrix with rows $r_1,\ldots,r_m$ and columns $c_1,\ldots,c_n$. For a branching $B$ of $D_M$, we define the \emph{$B$-split of $M$}, denoted by $M^B$, as the matrix with rows indexed by the elements of the set $U(B)$, and columns  $c'_1,\ldots,c'_n$, as follows.
Let $V = V(D_M)$ and for all $j\in \{1,\ldots, n\}$, let $v_j = \supp_M(c_j)$ (so $v_j\in V$).
For a vertex $v\in V$, we denote by  $B^+(v)$ the set of all vertices in $V$
reachable by a directed path from $v$ in $(V,B)$ (note that $v\in B^+(v))$.
For all $(r,v)\in U(B)$ and all $j\in \{1,\ldots, n\}$, set:
$$M^B_{(r,v),j}=\left\{
\begin{array}{ll}
1, & \hbox{if $v_j\in B^+(v)$;}  \\
0, & \hbox{otherwise.}
\end{array}
\right.
$$
\end{definition}

Note that if $M^B_{(r,v),j}=1$, then $r\in v_j$. See Fig.~\ref{fig:example} for an example of a binary matrix $M$ with two branchings $B_1$ and $B_2$ of its containment digraph and the corresponding row splits.

In the following lemma we show that the $B$-split of $M$ is a conflict-free row split of $M$ and compute the number of rows (resp., the number of distinct rows) of $M^B$.

\begin{lemma}\label{lemma:branching to split}
Let $M$ be a binary matrix without duplicated columns, $B$ a branching of $D_M$, and let $M^B$ be the $B$-split of $M$. Then $M^B$ is a conflict-free row split of $M$ with $|U(B)|$ rows, splitting each row $r_i$ of $M$ into rows of $M^B$ indexed by $U_B(r_i)$. Moreover, the number of distinct rows in $M^B$ is $|I(B)|$.
\end{lemma}

\begin{proof}
It is clear that the number of rows in $M^B$ is $|U(B)|$. For a row $r$ of $M$, we claim that $r$ is the bitwise OR of the rows of $M^B$ indexed by the set $U_B(r)$.
Suppose that $M_{r,j}=1$. Then $r\in v_j$.
We claim that there exists a vertex $v\in V$ such that $(r,v)\in U_B(r)$ and $M^B_{(r,v),j}=1$. If $(r,v_j)\in U_B(r)$, we can choose $v=v_j$ and
we are done. If this is not the case, then $r$ is covered in $v_j$, and hence $r\in v_k$ for some $v_k$ such that $(v_k,v_j)\in B$. Now if $(r,v_k)\not\in U_B(r)$, then we
repeat the argument with $v_k$ replaced by a ``covering'' in-neighbor. The procedure has to terminate after finitely many steps. Hence, we may assume that $(r,v_k)\in U_B(r)$.
This implies that $M^B_{(r,v_k),j}=1$.
Suppose now that $M_{r,j}=0$. Then $r\not \in v_j$ and therefore $M^B_{(r,v),j}=0$, for every $(r,v)\in U_B(r)$. This shows that $r$ is bitwise OR of the rows of $M^B$ indexed by $U_B(r)$, and therefore $M^B$ is row split of matrix $M$.

Suppose that two columns $c'_p$ and $c'_q$ of $M^B$ are in conflict.
Then there exist row indices, $(r_i,v_{i'}),(r_j,v_{j'})$ and $(r_k,v_{k'})$ in  $U(B)$ such that
$M^B_{(r_i,v_{i'}),p}=M^B_{(r_i,v_{i'}),q}=M^B_{(r_j,v_{j'}),p}=M^B_{(r_k,v_{k'}),q}=1$
and $M^B_{(r_j,v_{j'}),q}=M^B_{(r_k,v_{k'}),p}=0$.
Since $M^B_{(r_i,v_{i'}),p}=M^B_{(r_i,v_{i'}),q} = 1$, we have $v_p\in B^+(v_{i'})$ and $v_q\in B^+(v_{i'})$, that is,
$v_p$ and $v_q$ are reachable by a directed path from $v_{i'}$ in $(V,B)$.
Since $B$ is a branching, this is only possible if $v_q\in B^+(v_p)$ or $v_p\in B^+(v_q)$;
we may assume without loss of generality that $v_q\in B^+(v_p)$.
Since $M^B_{(r_j,v_{j'}),p}=1$, it follows that $v_p\in B^+(v_{j'})$. This further implies that $v_q\in B^+(v_{j'})$. Since $r_j\in v_{j'}$, it follows that $r_j\in v_q$, which contradicts the fact that $M^B_{(r_j,v_{j'}),q}=0$. The obtained contradiction shows that $M^B$ is conflict-free.

It remains to prove that the number of distinct rows in $M^B$ is $|I(B)|$.
Note that for any row $(r,v)$ in $M^B$ we have $v \in I(B)$.
Let $v\in I(B)$.
It is not difficult to see that for $r_i,r_j\in V$ such that $(r_i,v)\in U(B)$ and $(r_j,v)\in U(B)$, the rows of $M^B$ indexed by $(r_i,v)$ and $(r_j,v)$ are equal.
Hence the number of distinct rows in $M^B$ is at most $|I(B)|$.
To complete the proof we construct a set of size $|I(B)|$ of pairwise distinct rows of $M^B$. For every $v_i\in I(B)$, let $r^i$ be an arbitrary element of the (non-empty) set $v_i\setminus \cup N^-_B(v_i)$. Since $r^i$ is uncovered in $v_i$ with respect to $B$, the pair $(r^i,v_i)$ is an element of $U(B)$. We claim that the rows of $M^B$ indexed by $(r^i,v_i)$ over all $v_i\in I(B)$ are pairwise distinct.
Suppose that there exist $v_i$ and $v_j$ in $I(B)$ such that $v_i\neq v_j$ and the rows of $M^B$ indexed by $(r^i,v_i)$ and $(r^j,v_j)$ are equal. Since $M^B_{(r^i,v_i),i}=M^B_{(r^j,v_j),j}=1$ and the two rows are equal, we infer that
$M^B_{(r^i,v_i),j}=M^B_{(r^j,v_j),i}=1$.
Therefore $v_j\in B^+(v_i)$ and $v_i\in B^+(v_j)$. Since $B$ is a DAG, it follows that $v_i=v_j$, a contradiction.  This shows that there are exactly $|I(B)|$ distinct rows in $M^B$.
\end{proof}

The following lemma, exemplified in Fig.~\ref{fig:example3}, is the key to the converse direction.

\begin{lemma}\label{lemma:split to branching}
There exists an algorithm that takes as input a binary matrix $M$ without duplicated columns and a conflict-free row split $M'\in \{0,1\}^{m'\times n}$ of $M$, and computes in time ${\mathcal{O}}(m'n^2+n^{\omega})$ a branching
$B$ of $D_M$ such that $M^B$ can be obtained from $M'$ by removing some rows.
\end{lemma}

\begin{proof}
Denote the rows of $M$ with  $r_1,\ldots,r_m$ and the columns with $c_1,\ldots,c_n$. Let $R_i$ be the set of split rows of $r_i$, and let $c'_i$ be the column of $M'$ corresponding to $c_i$.
For $i\in \{1,\ldots,n\}$, let $v_i=\supp_M(c_i)$ and $v'_i=\supp_{M'}(c'_i)$.
We claim that for every $i,j\in \{1,\ldots,n\}$, if $(v'_i,v'_j)$ is an arc in $D_{M'}$ then $(v_i,v_j)$ is an arc in $D_M$.
Suppose that  $(v'_i,v'_j)$ is an arc in $D_{M'}$ and $(v_i,v_j)$ is not an arc in $D_M$. It follows that $v'_i\subseteq v'_j$ and $v_i\not \subseteq v_j$. Let $r_k\in v_i\setminus v_j$. Then $M'_{r',i}=1$ for some $r'\in R_k$. Since $v'_i\subseteq v'_j$, it follows that $M'_{r',j}=1$ and consequently $r_k\in v_j$, a contradiction.

\begin{figure}[h!]
\begin{center}
\includegraphics[width=\textwidth]{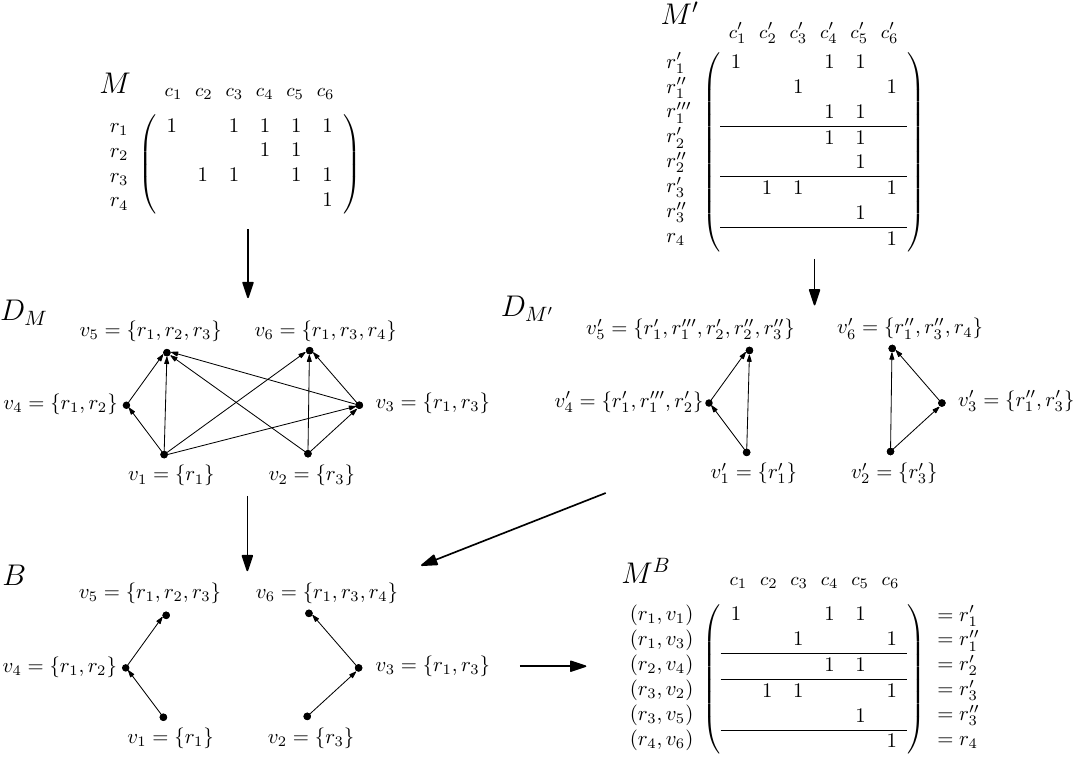}
\end{center}
\caption{An example of a binary matrix $M$, a conflict-free row split $M'$ of $M$, with their containment digraphs $D_M$ and $D_M'$, the resulting branching $B$ of $D_M$, and the resulting $B$-split $M^B$.}
\label{fig:example3}
\end{figure}

We say that an arc $(v'_i,v'_j)$ is {\em elementary} in $D_{M'}$ if there exists no $k\in \{1,\ldots,n\}$ such that both $(v'_i,v'_k)$ and $(v'_k,v'_j)$ are arcs of $D_{M'}$. Let $B$ be the subset of the arc set of $D_M$ defined by $(v_i,v_j)\in B$ if and only if $v'_i\neq \emptyset$ and $(v'_i, v'_j)$ is an elementary arc of $D_{M'}$.
We claim that $B$ is a branching of $D_M$.
Suppose that $(v_i,v_j)\in B$ and $(v_i,v_k)\in B$, for $j\neq k$.
Then, both $(v'_i, v'_j)$ and $(v'_i, v'_k)$ are elementary arcs of $D_{M'}$, which implies that $v_i' \subseteq v_j' \cap v_k'$.
Since $v'_i\neq \emptyset$ and $M'$ is conflict-free, it follows that $v_j' \subseteq v_k'$ or vice versa.
By definition of $B$, we obtain that $v_j' = v_k'$.
However, since $v_j\neq v_k$, we may assume that there exists some $r_p\in v_j\setminus v_k$, and therefore, there exists $r'\in R_p$, such that $r'\in v'_j$. Since $r_p\not \in v_k$ we have $R_p\cap v'_k=\emptyset$, contrary to the fact that $r'\in R_p\cap v'_j =R_p \cap v'_k$.
We conclude that $B$ is a branching.

Next, we prove that $M^B$ can be obtained from $M'$ by removing some rows, or, equivalently, that there exists a one-to-one mapping assigning to each row of $M^B$ an identical row of $M'$. Every row of $M^B$ is indexed by an element of $U(B)$.
Every element of $U(B)$ is of the form  $(r_i,v_k)$ with $(r_i,v_k)\in U_B(r_i)$ for some $i\in \{1,\ldots,m\}$ and $r_i\in v_k$.
To define a mapping as above, it suffices to show that there exists a row $r'$ of $M'$ such that $r'\in R_i$ and $r'$ is equal to the row of $M^B$ indexed by $(r_i,v_k)$, or more precisely  that $M'_{r',j}=1$ if and only if $v_j\in B^+(v_k)$.
First, observe that $r_i \in v_k$ implies that there exists some $r'\in R_i$ such that $M'_{r',k}=1$.

Assume that $M'_{r',j}=1$. Since  $M'_{r',k}=M'_{r',j}=1$ and $M'$ is conflict-free, it follows that either $v'_j\subseteq v'_k$ or $v'_k\subseteq v'_j$.
Suppose that $v'_j$ is a proper subset of $v'_k$ and therefore there exists a non-trivial $v'_j,v'_k$-path $P'$ consisting only of elementary arcs of $D_{M'}$. Since $r'$ is an element of $v'_j$, the set $v'_j$ is non-empty, which implies that the path $P'$ corresponds to a non-trivial $v_j,v_k$-path $P$ in $B$, therefore  $v_k\in B^+(v_j)$.
Since $M'_{r',j}=1$, and $r'\in R_i$, it follows that $r_i\in v_j$. However, this contradicts the fact that $(r_i,v_k)\in U_B(r_i)$. Therefore, $v'_k\subseteq v'_j$ and consequently $v_j\in B^+(v_k)$.
We proved that $M'_{r',j}=1$ implies that  $v_j\in B^+(v_k)$.

Suppose now that $v_j\in B^+(v_k)$. If $v_j=v_k$, then $j=k$ and $M'_{r',j}=M'_{r',k}=1$, as desired.
If $v_j\neq v_k$, then  since $v_j\in B^+(v_k)$, it follows that $v'_k\subseteq v'_j$. Combining this with the fact that $M'_{r',k}=1$, we conclude that $M'_{r',j}=1$. This completes the proof that the row $r'$ of $M'$ is equal to the row of $M^B$ indexed by $(r_i,v_k)$.

The above considerations imply the existence of a mapping assigning to each row of $M^B$ an identical row of $M'$. In fact, any mapping as defined above is also one-to-one, which can be seen as follows.
First, two rows of $M^B$ indexed by elements of $U(B)$ with distinct first coordinates, say $r_i$ and $r_j$, will be mapped to rows of $M'$ from $R_i$ and $R_j$, respectively, and by construction the sets $R_i$ and $R_j$ are disjoint.
Second, suppose we have two rows of $M^B$ indexed by elements of $U(B)$ with identical first coordinates but distinct second coordinates, say $(r_i,v_j)$ and $(r_i,v_k)$. The last part of the proof of Lemma~\ref{lemma:branching to split} implies that no two rows of $M^B$ indexed by pairs that differ in the values of their second coordinates are identical. Consequently, the images of rows of $M^B$ indexed by  $(r_i,v_j)$ and $(r_i,v_k)$ are also not identical (as binary vectors), and therefore they correspond to different rows of $M'$.

We conclude that $M^B$ can be obtained from $M'$ by deleting some rows.

It remains to estimate the time complexity of computing branching $B$. First, we compute the containment digraph $D_{M'}$ in time ${\mathcal{O}}(m'n^2)$.
Second, we compute the set $A'$ of elementary arcs of $D_{M'}$ in time ${\mathcal{O}}(n^{\omega})$ using the algorithm of Aho et al.~\cite{MR0306032}. Finally, branching $B$ can be computed from $A'$ in time ${\mathcal{O}}(|A'|) = {\mathcal{O}}(n^2)$. The claimed running time follows.
\end{proof}

Now we have everything ready to prove Theorem~\ref{thm:beta=gamma}.

\begin{sloppypar}
\begin{theorem1}
For every binary matrix $M\in \{0,1\}^{m\times n}$ with exactly $k$ distinct columns, the following holds:
\begin{enumerate}
  \item Any branching $B$ of $D_M$ can be transformed in time ${\mathcal{O}}(kmn)$ to a conflict-free row split of $M$ with exactly $|U(B)|$ rows and with exactly $|I(B)|$ distinct rows.

  \item Any conflict-free row split $M'\in \{0,1\}^{m'\times n}$ of $M$ can be transformed in time
  \hbox{${\mathcal{O}}(mn+m'k^2+k^{\omega})$} to a branching $B$ of $D_M$ such that $|U(B)|$ is at most the number of rows of $M'$ and
$|I(B)|$ is at most the number of distinct rows of $M'$.
\end{enumerate}
Consequently, for every binary matrix $M$, we have $\gamma(M)=\beta(M)$ and $\eta(M)=\zeta(M)$.
\end{theorem1}
\end{sloppypar}

\begin{proof}
Let $B$ be a branching of $D_M$. By Lemma~\ref{lemma:branching to split}, it suffices to show that
$M^B$, the $B$-split of $M$, can be computed in time ${\mathcal{O}}(kmn)$.
This can be achieved as follows. First, we compute the reduced matrix $\rd(M)$ in time ${\mathcal{O}}(mn)$ using radix sort.
Second, we compute the containment digraph $D_M$ in time ${\mathcal{O}}(k^2m) = {\mathcal{O}}(kmn)$ by performing pairwise comparisons of columns of $\rd(M)$.
Third, we compute the set $U(B)$ in time ${\mathcal{O}}(k^2m) = {\mathcal{O}}(kmn)$ by checking for each of the $k$ vertices $v\in V(D_M)$,
each of the ${\mathcal{O}}(m)$ elements $r\in v$, and each of the ${\mathcal{O}}(k)$ in-neighbors $u$ of $v$ in $D_M$ whether $r\in u$.
Fourth, in time ${\mathcal{O}}(k^2)$ we compute for each $v\in V(D_M)$ the set $B^+(v)$.
Finally, in time ${\mathcal{O}}(|U(B)|n)$ we compute the matrix $M^B$ using the definition. Note that
$|U(B)|\le km$, hence ${\mathcal{O}}(|U(B)|n) = {\mathcal{O}}(kmn)$ and the claimed time complexity follows.

Now, let $M'\in \{0,1\}^{m'\times n}$ be a conflict-free row split of $M$.

Consider first the case when $M$ is without duplicated columns.
In this case $k = n$ and by Lemma~\ref{lemma:split to branching}, in time ${\mathcal{O}}(m'n^2+
n^{\omega})$ a branching $B$ of $D_M$ can be computed such that $M^B$, the $B$-split of $M$, can be obtained from $M'$ by removing some rows. Since $|U(B)|$ (resp., $|I(B)|$) equals the number of rows (resp., the number of distinct rows) of $M^B$, this implies that $|U(B)|$ is at most the number of rows of $M'$ and $|I(B)|$ is at most the number of distinct rows of $M'$.

Consider now the general case. Let $X$ be a set of $k$ columns of $M$ such that the matrix $\rd(M)$ can be identified with the submatrix of $M$ obtained by considering only the columns in $X$. Let $M''\in \{0,1\}^{m'\times k}$ be the submatrix of $M'$ obtained by considering only the $k$ columns in $X$. Then, $M''$ is a conflict-free row split of $\rd(M)$. Note that matrix $M''$ can be computed in time proportional to its size, $m'k$, plus the number of columns, $n$, plus the time it takes to compute $\rd(M)$, which can be done in time ${\mathcal{O}}(mn)$ using radix sort. Since $\rd(M)$ is without duplicated columns, we have by the previous case
that in time ${\mathcal{O}}(m'k^2+k^{\omega}))$ a branching $B$ of $D_{\rd(M)} = D_{M}$ can be computed
such that $|U(B)|$ is at most the number of rows of $M''$ and
$|I(B)|$ is at most the number of distinct rows of $M''$.
Since the number of rows of $M''$ equals the number of rows of $M'$,
this immediately implies that $|U(B)|$ is at most the number of rows of $M'$.
Also, by construction, any two distinct rows of $M''$ correspond to a pair of distinct rows of $M'$ and hence
the number of distinct rows of $M''$ is at most the number of distinct rows of $M'$.
This implies that $|I(B)|$ is at most the number of distinct rows of $M'$.
The total time complexity is ${\mathcal{O}}(m'k+mn+m'k^2+k^{\omega})= {\mathcal{O}}(mn+m'k^2+k^{\omega})$, which establishes the second part of the theorem.

Finally, we show that $\gamma(M)=\beta(M)$ and $\eta(M)=\zeta(M)$.
Let $B$ be a branching of $D_M$ such that $|U(B)| = \beta(M)$.
By the first part of the theorem, there exists a conflict-free row split $M'$ of $M$ with $|U(B)|$ rows,
therefore $\gamma(M)\le |U(B)| = \beta(M)$.
Conversely, if $M'$ is a conflict-free row split of $M$ with $\gamma(M)$ rows, then there exists a
branching $B$ of $D_M$ such that $|U(B)|$ is at most the number of rows of $M'$ (that is, $\gamma(M)$).
This implies $\beta(M)\le |U(B)|\le \gamma(M)$. Therefore, $\gamma(M)=\beta(M)$.
The proof of equality $\eta(M)=\zeta(M)$ is analogous.
\end{proof}

\section{A Strengthening of Dilworth's Theorem and its Connection to the
 Minimum Conflict-Free Row Split Problem}\label{sec:Dilworth}

By Theorem~\ref{thm:beta=gamma}, the MCRS problem can be concisely formulated in terms of a problem on branchings in a derived digraph. As shown by Hujdurovi\' c et al.~in~\cite{tcbb16}, the MCRS problem is \np-hard; consequently, the MUB problem is also \np-hard. In this section we show that a related problem in which we examine only a subset of all the branchings of the containment digraph of the input binary matrix is polynomially solvable. This will be achieved by deriving, in Section~\ref{sec:min-max}, a min-max theorem generalizing the classical Dilworth's theorem on partially ordered sets, which may be of independent interest. The resulting heuristic algorithm will be described in Section~\ref{sec:heuristic} (see also Remark~\ref{remark:tight} on p.~\pageref{remark:tight}).

\subsection{A Min-Max Relation Strengthening Dilworth's Theorem}\label{sec:min-max}

This section can be read independently of the rest of the paper.

Consider a pair $(D,\smallPi)$ where $D=(V,A)$ is a DAG
and $\smallPi : V \to \Zplus$ is a \emph{weight function} of $D$. (We use $\Zplus$ for the set of non-negative integers.)
The weight function $\smallPi$ is called \emph{monotone} if $\smallPi_u \leq \smallPi_v$
for every $(u,v)\in A$.

In $D$, a \emph{non-trivial path} is a directed path with at least one arc.
We denote by $D^t$ the transitive closure of $D$, that is, the DAG $(V,A^t)$ on the same vertex set as $D$ having an arc $(u,v)\in A^t$ if and only if there exists a non-trivial path in $D$ from $u$ to $v$.
A chain in $D$ is a sequence of vertices $C = (v_1, v_2, \ldots, v_s)$
such that $(v_i,v_{i+1})\in A^t$ for all $i\in \{1,\ldots, s-1\}$;
sometimes we regard $C$ as the set of its vertices $C = \{v_1, v_2, \ldots, v_s\}$.
The \emph{price of chain} $C$ is given by
$\bigPi(C) = \max_{v\in C} \smallPi_v$.
A family of vertex-disjoint chains $P = \{C_1, \ldots, C_p\}$
is called a \emph{chain partition} of $D$ if every vertex of $D$ is contained
in precisely one chain of $P$.
The \emph{price} of chain partition $P$ is defined as
$\bigPi(P) = \sum_{i=1}^p \bigPi(C_i)$.
Consider the following problem.

\medskip
\begin{center}
\fbox{\parbox{0.98\linewidth}{\noindent
{\sc MinimumPriceChainPartition:}\\[.8ex]
\begin{tabular*}{.95\textwidth}{rl}
{\em Input:} & A DAG $D = (V,A)$ and a monotone weight function $\smallPi : V \to \Zplus$ of $D$.\\
{\em Task:} & Compute a chain partition $P$ of $D$ such that the price $\bigPi(P)$  is minimum possible.
\end{tabular*}
}}
\end{center}
\medskip

In this section we give a polynomial-time algorithm and a min-max
characterization for the above problem.
As can be expected, the notion of antichain
will play a main role in this min-max characterization.
An \emph{antichain} of $D$ is a set of vertices $N\subseteq V$
such that $N$ is an independent set (that is, a set of pairwise non-adjacent vertices) in $D^t$;
in other words, no non-trivial path of $D$ has both endpoints in $N$.
Note that $|C\cap N| \leq 1$
for any chain $C$ and any antichain $N$.
The \emph{width} of $D$,
denoted by ${\it wdt}(D)$, is the maximum cardinality of an antichain in $D$.

A classical theorem of Dilworth states that ${\it wdt}(D)$
equals the minimum number of chains in a chain partition of $D$~\cite{Dilworth}.
Moreover, a chain partition of $D$ into ${\it wdt}(D)$ chains can be computed in time $\widetilde{\mathcal{O}}(n^\omega)$ where $n = |V(D)|$, $\omega$ is any real number such that there exists an ${\mathcal{O}}(n^{\omega})$ algorithm for multiplying two $n\times n$ binary matrices (e.g., $\omega=2.373$), and the $\widetilde{\mathcal{O}}(\cdot)$ notation ignores logarithmic factors. Indeed, by applying the approach of Fulkerson~\cite{MR0078334} (see also~\cite{MR545530,makinen2015genome}), a minimum chain partition of $D$ can be computed by solving a maximum matching problem in a derived bipartite graph having $2n$ vertices. This can be done in time $\widetilde{\mathcal{O}}(n^\omega)$ using the algorithm of~\cite{MR636312}.\footnote{Alternatively, one could use the bipartite matching algorithm from~\cite{MR1356505} to obtain the (incomparable) running time of ${\mathcal O}(\sqrt{n}m\log_n(n^2/m))$ where $m=|A^t|$ is the number of edges in the transitive closure of $D$. For the sake of simplicity of presentation, we state the theorem with the running time resulting from using the Ibarra-Moran algorithm.}

For later use, we summarize these facts as follows.

\begin{theorem}[Dilworth's theorem]\label{thm:Dilworth}
Every DAG $D$ admits a chain partition of size ${\it wdt}(D)$.
Such a chain partition can be computed in time
$\widetilde{\mathcal{O}}(|V(D)|^\omega)$.
\end{theorem}

Our characterization will be a refinement of Dilworth's theorem
and its algorithmic proof makes use of Dilworth's theorem as a subroutine.
We must introduce one further notion however.
A \emph{tower} of antichains of $D$
is a sequence of antichains of $D$,
$T = (N_1, N_2, \ldots, N_{{\it wdt}(D)} )$,
with $|N_i| = i$.
The \emph{value of an antichain} $N$ is given by
${\it val}(N) = \min_{v\in N} \smallPi_v$
and the \emph{value} of tower $T = (N_1, N_2, \ldots, N_{{\it wdt}(D)} )$
is defined as
${\it val}(T) = \sum_{i=1}^{{\it wdt}(D)} {\it val}(N_i)$.

To appreciate the purpose of this notion, we begin with a simple observation.

\begin{sloppypar}
\begin{lemma}\label{lem:theLowerBound}
Let $D$ be a DAG, let $P = \{C_1, \ldots, C_p\}$ be a chain partition of $D$,
and let $T = (N_1, N_2, \ldots, N_{{\it wdt}(D)})$ be a tower of antichains of $D$.
Then, $\bigPi(P) \geq {\it val}(T)$ even if the weight function $\smallPi$ is not monotone.
\end{lemma}
\end{sloppypar}

\begin{proof}
For every chain $C$ and every antichain $N$
we have that $|C\cap N| \leq 1$.
Moreover, if $|C\cap N| = 1$,
then $\bigPi(C) \geq {\it val}(N)$.
Indeed, if $C\cap N = \{z\}$, then
\[
   \bigPi(C) =\max_{v\in C} \smallPi_v \geq \smallPi_z \geq \min_{v\in N} \smallPi_v = {\it val}(N).
\]
Since $P$ is a chain partition of $D$, then $|P| \geq {\it wdt}(D)$,
and we can always rename its chains as
$\tilde{C}_1, \tilde{C}_2, \ldots, \tilde{C}_p$
in such a way that,
for every $i=1,\ldots, {\it wdt}(D)$,
chain $\tilde{C}_i$ intersects the antichain $N_i$.
At this point,
\[
   \bigPi(P) = \sum_{i=1}^p \bigPi(C_i)
          = \sum_{i=1}^p \bigPi(\tilde{C}_i)
          \geq \sum_{i=1}^{{\it wdt}(D)} \bigPi(\tilde{C}_i)
          \geq \sum_{i=1}^{{\it wdt}(D)} {\it val}(N_i)
          = {\it val}(T).
\]
\end{proof}

For the case of monotone weight functions, the following min-max strengthening of Dilworth's theorem holds.

\begin{theorem}\label{thm:min-max}
Let $D$ be a DAG and let $\smallPi$ be a monotone weight function of $D$.
Then $D$ admits a chain partition $P = \{C_1, \ldots, C_{{\it wdt}(D)}\}$
and a tower of antichains $T = (N_1, N_2, \ldots, N_{{\it wdt}(D)} )$
such that $\bigPi(P) = {\it val}(T)$.
Such a pair $(P,T)$ can be computed in time
$\widetilde{\mathcal{O}}(|V(D)|^{\omega+1})$.
\end{theorem}

\begin{proof}
The proof is by induction on $n= |V(D)|$.
Clearly, the statement holds for $n=1$.
As for the inductive step, let $n>1$ and consider a vertex $v\in V(D)$ without any incoming arcs and such that
$\smallPi_v\le \smallPi_{v'}$ for all $v'\in V(D)$.
Such a vertex exists since the subgraph of $D$ induced by the set of vertices achieving the minimum
value of $\smallPi$ is acyclic.
Let $D' = D-v$, and consider a chain partition
$P' = \{C_1, \ldots, C_{{\it wdt}(D')}\}$
of $D'$ and a tower of antichains
$T' = (N_1, N_2, \ldots, N_{{\it wdt}(D')})$
of $D'$ such that $\bigPi(P') = {\it val}(T')$.

Two cases are possible.
If ${\it wdt}(D) > {\it wdt}(D')$
then let $P$ be obtained from $P'$
by adding a chain $C$ comprising the sole vertex $v$
and let $T$ be obtained from $T'$
by adding any antichain $N_{{\it wdt}(D)}$
of $D$ such that $|N_{{\it wdt}(D)}| = {\it wdt}(D)$.
Since $|N_{{\it wdt}(D)}| > {\it wdt}(D')$, we infer that $v\in N_{{\it wdt}(D)}$
and hence $\min_{u\in N_{{\it wdt}(D)}} \smallPi_u = \smallPi_{v} = \bigPi(C)$.
Therefore, $\bigPi(P) = {\it val}(T)$ closing the induction in this case.

Assume therefore that ${\it wdt}(D) = {\it wdt}(D')$.
Let $T$ be an antichain in $D'$ with $|T| = {\it wdt}(D)$
and let $\widehat{T}$ be the set of vertices of $D$
from which there is a non-trivial path to a vertex of $T$.
Notice that $v \in \widehat{T}$
since $T\cup \{v\}$ is not an antichain and $v$ is a source vertex of $D$.
The DAG $D[V(D)\setminus \widehat{T}]$ is an acyclic subgraph of $D'$ of width at least $|T| = wdt(D')$, since $T$ is a subset of its vertex set. Moreover, while the width of an arbitrary induced subgraph can in general increase with vertex removal, this is not the case for $D[V(D)\setminus \widehat{T}]$, because any path in $D$ between two vertices of $V(D)\setminus \widehat{T}$ is also a path in $D[V(D)\setminus \widehat{T}]$, by the choice of $\widehat{T}$. It follows that the DAG $D[V(D)\setminus \widehat{T}]$ is of width $|T| = {\it wdt}(D')$; hence, by the inductive hypothesis, it
admits a chain partition $P^T = \{C^T_1, \ldots, C^T_{{\it wdt}(D')}\}$ with $\bigPi(P^T)\leq \bigPi(P') = {\it val}(T')$. (Indeed, we could just take $C^T_i := C_i\setminus \widehat{T}$ for every $i\in\{1,\ldots, {\it wdt}(D')\}$.)
Also the acyclic subgraph $D[T\cup \widehat{T}]$ of $D$ has width $|T| = {\it wdt}(D')$;
hence, by Dilworth's theorem it admits a chain partition
$P^{\widehat{T}} = \{C^{\widehat{T}}_1, C^{\widehat{T}}_2, \ldots, C^{\widehat{T}}_{{\it wdt}(D')}\}$ covering all its vertices.
Now we construct our chain partition for $D$:
let $T = \{t_1, t_2, \ldots, t_{|T|}\}$.
After a possible renaming of the chains in the two chain partitions,
we can assume that
$C^T_i\cap T = \{t_i\}$ and $C^{\widehat{T}}_i\cap T = \{t_i\}$
for every $i=1,\ldots, {\it wdt}(D') = {\it wdt}(D) = |T|$,
and hence define the chain $\tilde{C}_i = C^{\widehat{T}}_i\cup C^T_i$.
(Indeed, $t_i$ will be the last vertex of $C^{\widehat{T}}_i$
and the first vertex of $C^{T}_i$, thus this chaining of chains can be performed.)
Note that $\tilde{P} := \{\tilde{C}_1, \ldots, \tilde{C}_{{\it wdt}(D)}\}$
is a chain partition of $D$
with $\bigPi(\tilde{P}) = \bigPi(P^T)\leq \bigPi(P') = {\it val}(T')$
Clearly, $T'$ is a valid tower of antichains for $D$.

The above proof of the existence of a pair $(P,T)$ of a chain partition and a tower of antichains of $D$ satisfying $\bigPi(P) = {\it val}(T)$
is constructive and can be turned into a
$\widetilde{\mathcal{O}}(|V(D)|^{\omega+1})$
time algorithm for computing such a pair $(P,T)$.
Indeed, at each step of the algorithm, we delete one vertex, make one recursive call to the algorithm, compute the set $\widehat T$ and the acyclic subgraph $D[T\cup \widehat{T}]$ together with a chain partition $P^{\widehat{T}} = \{C^{\widehat{T}}_1, C^{\widehat{T}}_2, \ldots, C^{\widehat{T}}_{{\it wdt}(D')}\}$ covering all its vertices. The time complexity of each step is dominated by computing $P^{\widehat{T}}$. By Theorem~\ref{thm:Dilworth}, this can be done in time
$\widetilde{\mathcal{O}}(|V(D)|^\omega)$.
The claimed time complexity of
$\widetilde{\mathcal{O}}(|V(D)|^{\omega+1})$
follows.
\end{proof}

To see that Theorem~\ref{thm:min-max} is a strengthening of Dilworth's theorem, consider an arbitrary DAG $D=(V,A)$ and let $\smallPi$ be the weight function of $D$ that is constantly equal to~$1$. Then, the price of any chain $C$ is $\bigPi(C) = \max_{v\in C} \smallPi_v = 1$ and the price of a chain partition $P$ equals its cardinality. Moreover, the value of any antichain $N$ is ${\it val}(N) = \min_{v\in N} \smallPi_v = 1$, and consequently the value of any tower $T = (N_1, N_2, \ldots, N_{{\it wdt}(D)})$ of antichains is ${\it val}(T) = \sum_{i=1}^{{\it wdt}(D)} {\it val}(N_i) = {\it wdt}(D)$. Since ${\it wdt}(D)$ is a lower bound on the cardinality of any chain partition, applying Theorem~\ref{thm:min-max} to $(D,\smallPi)$ gives exactly the statement of Dilworth's theorem for $D$.

We would also like to emphasize that due to the non-linearity of the definitions of the price of a chain and the value of an antichain, Theorem~\ref{thm:min-max} is incomparable with the classical weighted generalization of Dilworth's theorem due to Frank~\cite{MR1714683}.

\begin{sloppypar}
\begin{remark}
The monotonicity assumption in Theorem~\ref{thm:min-max} is necessary.
If we drop it, the price $\min_{P} \bigPi(P)$ and the value $\max_{T} {\it val}(T)$ may diverge.
Consider the DAG on vertex set $\{a\}, \{b\}, \{a,b\}, \{b,c\}$,
  with $\smallPi_{\{a\}} = \smallPi_{\{b,c\}} = z$
  and $\smallPi_{\{b\}} = \smallPi_{\{a,b\}} = Z$, with $0<z < Z <2z$,
  and where the arcs are according to set inclusion.
  Here, $\min_{P} \bigPi(P) = 2Z$ whereas $\max_{T} {\it val}(T) = Z+z$.

On the other hand, a simple application of Dilworth's theorem shows that
the monotonicity assumption is not necessary in the case of $0,1$-weight functions.
\end{remark}
\end{sloppypar}

Lemma~\ref{lem:theLowerBound} and Theorem~\ref{thm:min-max} imply the following.

\begin{corollary}\label{cor:min-price-chain-partition-algorithm}
{\sc MinimumPriceChainPartition} can be solved optimally in time
$\widetilde{\mathcal{O}}(|V(D)|^{\omega+1})$.
More specifically, in the stated time a minimum price chain partition $P$ of $D$ can be found
with the additional property that $|P| = {\it wdt}(D)$ (hence $P$ is simultaneously a
minimum price chain partition and a minimum size chain partition of $D$).
\end{corollary}

Two remarks are in order here, showing that the result of
Corollary~\ref{cor:min-price-chain-partition-algorithm} is sharp in two ways.
First, let us note that the variant of the {\sc MinimumPriceChainPartition}
problem in which the chains used in the partition have to be of bounded size
was studied by Moonen and Spieksma~in~\cite{MR2386521}, who described a
practical application encountered at Bruynzeel Storage Systems, a manufacturing company in the
Netherlands, to a problem of optimally loading pallets on a truck.\footnote{The upper bound on the size of chains
relates to the fact that trucks are of bounded height.}
Moonen and Spieksma referred to the problem as ``Minimum Weight Partition into
$B$-chains'' (where $B$ is the upper bound on the size of the chains)
and showed that the problem is \apx-hard even in the case of unit
weights, strengthening the previous \np-hardness result from~\cite{MR1380086}.

Second, the variant of {\sc MinimumPriceChainPartition} where the weight
function $\smallPi$ is not restricted to be monotone is \np-hard. This follows
from the fact that the Weighted Coloring problem is \np-hard in the class of
interval graphs, as shown by Escoffier et al.~\cite{MR2195353}.
The input to the Weighted Coloring problem is a graph $G = (V,E)$ and a weight function
$\smallPi:V \to \Zplus$ and the task is to find a partition $\mathcal{I}$ of $V$ into
independent sets minimizing the value of $\sum_{I\in \mathcal{I}}\max_{v\in
I}\smallPi_v$. The Weighted Coloring problem in interval graphs finds applications in
distributed computing in transportation networks and in dynamic storage allocation
in computer processes~\cite{MR1439871}. Given an interval graph $G=(V,E)$ represented by an interval model
$(I_v = [a_v,b_v]:v\in V)$ and a weight function $\smallPi:V \to \Zplus$, the
Weighted Coloring problem given $(G,\smallPi)$ is equivalent to the problem of
finding a chain partition of the DAG with vertex set $V$ and arc set
$\{(u,v): b_u<a_v\}$ of minimum price with respect to $\smallPi$. The claimed
\np-hardness follows.

\subsection{Connection with the Minimum Conflict-Free Row Split Problem}\label{sec:heuristic}

We will now describe a heuristic algorithm for the MCRS problem based on Theorem~\ref{thm:min-max} and its algorithmic proof. The basic idea is to search for an optimal solution only among linear branchings, where a branching of $D_M$ is said to be {\it linear} if it defines a subgraph of maximum in- and out-degree at most one, that is, a disjoint union of directed paths. Note that such branchings correspond bijectively to chain partitions of $D_M$.

We denote with $\beta_{\ell}(M)$ the minimum number of elements in $U(B)$ over all linear branchings $B$ of $D_M$. We now introduce the following problem, referred to as {\sc MinimumUncoveringLinearBranching}: Given a binary matrix $M$, compute a linear branching $B$ of $D_M$ such that  $|U(B)|=\beta_{\ell}(M)$.

For a binary matrix $M$, define a function $\smallPi: V(D_M)\to \Zplus$ with $\smallPi(v)=|v|$ (recall that vertices of $D_M$ are pairwise distinct subsets or $R_M$).
By definition of $D_M$, we have $u\subset v$ whenever $(u,v)$ is an arc in $D_M$. This implies that $\smallPi$ is a monotone weight function of $D_M$.
It is not difficult to see that for a linear branching $B$ and its corresponding chain partition $P$, we have $\bigPi(P)=|U(B)|$.
Since linear branchings correspond bijectively to chain partitions, it follows that {\sc MinimumUncoveringLinearBranching} is a special case of {\sc MinimumPriceChainPartition}. Using Theorem~\ref{thm:min-max}, we obtain that a linear branching $B$ of $D_M$ with $|U(B)|=\beta_{\ell}(M)$ can be computed in time $\widetilde{\mathcal{O}}(|V(D)|^{\omega+1})$.
This proves the following theorem.

\begin{sloppypar}
\begin{theorem}\label{thm:linear-branchings}
{\sc MinimumUncoveringLinearBranching} can be solved to optimality in time
$\widetilde{\mathcal{O}}(|V(D)|^{\omega+1})$.
\end{theorem}
\end{sloppypar}

Note that Theorem~\ref{thm:linear-branchings} yields a heuristic polynomial-time algorithm for
the MUB problem, and consequently for the MCRS problem. We are now going to explain why this algorithm improves
on the heuristic for the latter problem by Hujdurovi\'c et al.~from~\cite{tcbb16}. For the sake of
simplicity of exposition, suppose that the input matrix $M$ does not have any pairs of identical columns.
(It is not difficult to see that this assumption is without loss of generality.) In this case,
the algorithm from \cite{tcbb16} returns a row split of the input matrix naturally derived from
an optimal coloring of the complement of the underlying undirected graph of $D_M$, which is a
cocomparability graph and thus an optimal coloring can be computed efficiently, see, e.g., \cite{MR2063679}.
Such optimal colorings correspond bijectively to minimum chain partitions of $D_M$; each color class
corresponds to a chain. In the terminology of branchings, the conflict-free row split of the input matrix $M$ returned by the heuristic from~\cite{tcbb16} is exactly the $B$-split of $M$ (cf.~Definition~\ref{def:B-split}) where $B$ is the linear branching of $D_M$ corresponding to a minimum chain partition of $D_M$.

In the above approach, any proper coloring could be used instead of an optimal coloring of the derived cocomparability graph. In branching terminology, choosing a proper coloring of the derived cocomparability graph so that the number of rows of the output row split is minimized corresponds exactly to {\sc MinimumUncoveringLinearBranching}, which can be solved optimally by Theorem~\ref{thm:linear-branchings}. Thus, the heuristic algorithm for the MCRS problem that returns the $B$-split of $M$ where $M$ is an optimal solution to {\sc MinimumUncoveringLinearBranching} always returns solutions that are at least as good as those computed by the algorithm by Hujdurovi\'c et al.~from~\cite{tcbb16}.
Moreover, note that by Corollary~\ref{cor:min-price-chain-partition-algorithm}, digraph $D_M$ has a minimum price chain partition that is also minimum with respect to size. This implies the existence of an optimal solution to {\sc MinimumUncoveringLinearBranching} on $M$ such that the corresponding chain partition is of size ${\it wdt}(M)$ and, equivalently, the existence of an optimal coloring of the derived cocomparability graph that minimizes the number of rows in the derived conflict-free row split of $M$ over all proper colorings of the derived graph.

\begin{remark}
As discussed in~\cite{wabi14,tcbb16}, the main motivation for the MCRS problem comes from cancer genomics, with the goal to reconstruct, from a set of given mixed tumor samples,
a simplest possible mutational history of the tumor, represented by a rooted tree (without any restriction on the shape of the tree). Without going into details, let us note that the output of the heuristic algorithm for the MCRS problem given by Theorem~\ref{thm:linear-branchings} corresponds to a simplest possible reconstruction of the mutational history within a restricted space of rooted trees, namely within the space of rooted trees such that the root is the only node that is allowed to have more than one non-leaf child.
\end{remark}


\section{(In)approximability Issues}\label{sec:approx}

In this section we will discuss (in)approximability properties of the four problems studied in this paper, giving both
\apx-hardness results and approximation algorithms. The approximation ratios of some of our algorithms will be described in terms of the following parameters of the input matrix. Recall that the \emph{width} of a DAG $D$ is the maximum cardinality of an antichain in $D$. The {\em height} of a DAG $D$ is the maximum number of vertices in a directed path contained in $D$. The {\it width} and the {\it height} of a binary matrix $M$ are denoted by ${\it wdt}(M)$ and by $h(M)$, respectively, and defined as the width, resp.~the height, of the containment digraph of $M$.

\subsection{Hardness Results}\label{subsec:hard}

Our main inapproximability results are summarized in the following theorem, which shows hardness
already for very restricted input instances.

\begin{sloppypar}
\begin{theorem}\label{thm:apxHard}
The MUB and the MIB problems (and consequently the
MCRS and the MDCRS problems) are \apx-hard, even for instances of height~$2$.
\end{theorem}
\end{sloppypar}

The above result implies that none of the four problems admits a polynomial-time approximation scheme (PTAS), unless {\sf P} = \np.
Proving that a problem is \apx-hard also provides a different proof of \np-hardness.

The \apx-hardness for the two branching problems is established by developing $L$-reductions from the vertex cover problem in cubic graphs, which is known to be \apx-hard~\cite{MR1756204}. The \apx-hardness of the other two problems then follows from Theorem~\ref{thm:beta=gamma}. Recall that $\apx$ is a class of problems approximable to within a constant factor in polynomial time. A problem $\Pi$ is said to be \apx-hard if every problem in $\apx$ reduces to $\Pi$ by an approximation-preserving reduction. Another way to prove that a problem $\Pi$ is \apx-hard is to show that an \apx-complete problem $\Pi'$ is $L$-reducible to $\Pi$.
For the sake of self-containment, we recall the definition of $L$-reducibility; for further background on \apx-hardness, we refer to~\cite{complexity}.

\begin{definition}\label{def:L-reduction}
Let $\Pi$ and $\Pi'$ be two \np-hard optimization problems. Problem $\Pi$ is said to be {\em $L$-reducible} to problem $\Pi'$ if there exists a polynomial-time transformation $f$ mapping instances of $\Pi$ to instances of $\Pi'$ and constants $a,b \in \mathbb{R}_+$ such that for every instance $x$ of $\Pi$ the following conditions hold:
\begin{itemize}
\item $opt_{\Pi'} (f(x)) \le a\cdot opt_{\Pi} $,

\item for every feasible solution $y'$ of $f(x)$ with objective value $c_2$ we can compute in polynomial time solution $y$ for $x$ with objective value $c_1$ such that $|opt_{\Pi}(x) - c_1| \le b \cdot |opt_{\Pi'}(f(x)) - c_2|$.
\end{itemize}
\end{definition}

To simplify the description of the hardness reductions of this section, we will use the notion of a column hypergraph of a given binary matrix $M$.
This notion is closely related to the containment digraph of $M$ and will find a further application in Section~\ref{sec:2-approx}.
Recall that a {\em set family} (or a {\em hypergraph}) is a pair $\mathcal{H} = (V,\mathcal{E})$ where $V= V(\mathcal{H})$ is a set and $\mathcal{E} = E(\mathcal{H})$ is a subset of the power set $\mathcal{P}(V)$. Elements of $V(\mathcal{H})$ are the {\em vertices} of $\mathcal{H}$; elements of $E(\mathcal{H})$ are its {\em hyperedges}. The {\it column hypergraph} $\mathcal{H}_M$ of a binary matrix $M$ is the hypergraph having the rows of $M$ as vertices and the support sets of the columns of $M$ as hyperedges. Formally, $\mathcal{H}_M$ has vertex set $V(\mathcal{H}_M) = R_M$ and hyperedge set $E(\mathcal{H}_M) = \{\supp_M(c): c\in C_M \}$.
Note that the set of hyperedges of the column hypergraph of $M$ equals the vertex set of the containment digraph $D_M$.

We split the proof of Theorem~\ref{thm:apxHard} into two parts.

\begin{proposition}\label{prop:apxHard}
{\sc MinimumUcoveringBranching} is \apx-hard, even for instances of height~$2$.
Consequently, {\sc MinimumConflict-FreeRowSplit} is \apx-hard, even for instances of height~$2$.
\end{proposition}

\begin{proof}
We will prove the proposition using the fact that the vertex cover problem is \apx-hard on cubic graphs~\cite{MR1756204}. Recall that a graph $G$ is {\it cubic} if every vertex of $G$ is incident with exactly three edges and that a vertex cover of a graph $G$ is a subset $C \subseteq V(G)$ such that $\forall \, e=\{v_1,v_2\}, \, e \in E(G) \Rightarrow v_1 \in C \, \vee \, v_2 \in C$.
For all $v\in V(G)$, we define $E(v)$ as a set of all edges in $E(G)$ incident with $v$. In symbols, $E(v) = \{e \in E(G) : v \in e \}$. We say that a graph $G$ is {\em cubic} if for every $v\in V(G)$ it holds $|E(v)| = 3$.

 We will construct an $L$-reduction from the vertex cover problem in cubic graphs to the MUB problem on instances of height~$2$. Let $G$ be a cubic graph.  Let $x$ and $y$ be two new vertices not in $V(G) \cup E(G)$. Let $R=E(G) \cup \{x,y\}$ and let $\mathcal{H}$ be the hypergraph with vertex set $R$ and edge set
\begin{eqnarray*}
\mathcal{E} &=& \{E(G)\cup \{x\}\}
\cup \{ E(v) \cup \{x\} : v\in V(G) \}
\cup \{ E(v) \cup \{y\} : v\in V(G) \}\\
&&{ } \cup \{ E(v) \cup \{x,y\} : v\in V(G)\} \,.
\end{eqnarray*}
Let $M$ be a binary matrix without duplicated columns such that the column hypergraph of $M$ is isomorphic to $\mathcal{H}$. Note that $M$ is of height~$2$. See Fig.~\ref{fig:construction} for an example construction, representing the containment digraph $D_M$ of the binary matrix derived from the complete graph $K_4$.
\begin{figure}[H]
\begin{center}
\includegraphics[width=1\textwidth]{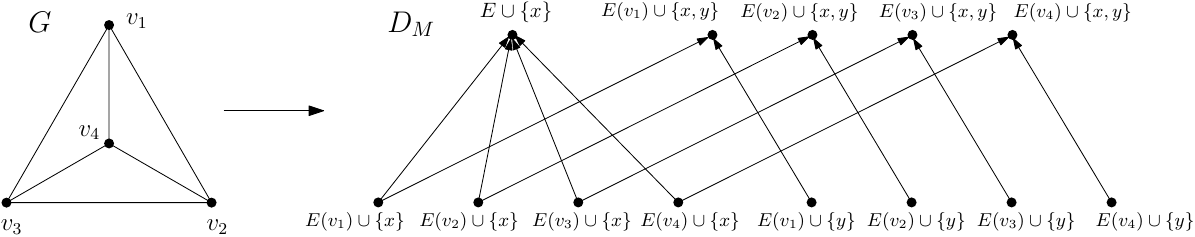}
\end{center}
\caption{An example construction of the $L$-reduction from the proof of Proposition~\ref{prop:apxHard}.}
\label{fig:construction}
\end{figure}

We denote by $\tau(G)$ the vertex cover number of $G$, that is, the minimum size of a vertex cover in $G$.
The \apx-hardness of {\sc MinimumUcoveringBranching} will be a consequence of the following claim and its proof.

\medskip
\noindent{\bf Claim. $\tau(G)  = \beta(M)-8|V(G)|$.}

\begin{proof}[Proof of the claim]
We split the proof of the equality into two parts, proving each of the two inequalities separately.

First, we prove the inequality $\beta(M)\le \tau(G)+ 8|V(G)|$.
Let $C$ be a minimum vertex cover of $G$. Define a branching $B$ of $D_M$ as follows:
\begin{eqnarray*}
B & = &
\{(E(v)\cup \{y\}, E(v)\cup \{x,y\}) : v \in V(G) \} \\
&& { } \cup  \{(E(v)\cup \{x\}, E\cup \{x\}) : v \in C \} \\
&& { } \cup  \{(E(v)\cup \{x\}, E(v)\cup \{x,y\}) : v \in V(G)\setminus C\}
\end{eqnarray*}
See Fig.~\ref{fig:construction1} for an example.

\begin{figure}[h!]
\begin{center}
\includegraphics[width=0.8\textwidth]{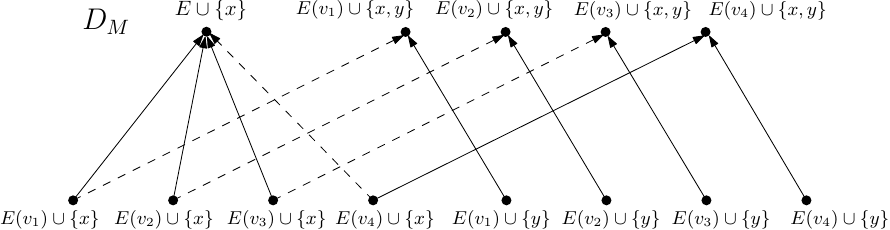}
\end{center}
\caption{The set of non-dashed arcs corresponds to the branching obtained from the vertex cover $\{v_1,v_2,v_3\}$.}
\label{fig:construction1}
\end{figure}

It is clear from the construction that $B$ is indeed a branching. Since $C$ is a vertex cover, every $e \in E(G)$ is covered  in $E(G)\cup \{x\}$ with respect to $B$. It is now not difficult to see the set of uncovered pairs with respect to $B$ equals
\begin{eqnarray*}
U(B)  &=& \{(r,E(v)\cup \{z\}) : v \in V(G) , z \in \{x,y\} , r \in E(v) \cup \{z\} \}\\&&{ }  \cup  \{(x,E(v)\cup \{x,y\}) : v \in C \} \,.
\end{eqnarray*}
Since we have $|E(v)| = 3$ for all $v\in V(G)$, this implies $\beta(M) \le |U(B)| = 8|V(G)|+ |C|= 8|V(G)|+\tau(G)$, as claimed.

Now we prove the inequality $ \tau(G) \le \beta(M) - 8|V(G)|$.
Let $B$ be a branching of $D_M$ such that $|U(B)|=\beta(M)$.
For every source vertex $u$ in $D_M$ and every element $r \in u$ it holds that $r$ is uncovered in $u$.
Since the source vertices are exactly the vertices of the form $E(v) \cup \{x\}$ and $E(v) \cup \{y\}$,
we have exactly $8|V(G)|$ uncovered pairs corresponding to the source vertices.
The minimality of $B$ implies that all arcs of the form $(E(v)\cup \{y\},E(v)\cup \{x,y\})$ are in $B$.
Therefore, for every $v\in V(G)$, element $x$ is the only possibly uncovered element in vertex $E(v)\cup \{x,y\}$.

We show that we may assume that vertex $E(G) \cup \{x\}$ is not irreducible, that is, that all its elements are covered in $E(G)\cup \{x\}$.
Suppose first that $x$ is not covered in $E(G)\cup \{x\}$. Then $B$ does not contain any arc of the form $(E(v)\cup \{x\}, E(G)\cup \{x\})$, and therefore, by minimality, contains all arcs of the form $(E(v)\cup \{x\}, E(v)\cup \{x,y\})$. Replacing one of these arcs with the arc
$(E(v)\cup \{x\}, E(G)\cup \{x\})$ results in a branching $B'$ such that $|U(B')| \le |U(B)|$,
hence in an optimal branching covering $x$. Now, suppose that there exists some $e\in E(G)$ such that $e \not \in \cup N^-_B(E(G) \cup \{x\})$. Let $v$ be an endpoint of $e$ in $G$ and consider the vertex $E (v) \cup \{x\}$. Since $e$ is not covered in $E(G)\cup \{x\}$, the arc $(E(v)\cup \{x\},E(G)\cup \{x\})$ is not in $B$. The optimality of $B$ implies that $(E(v)\cup \{x\},E(v)\cup \{x,y\}) \in B$.
Now, replace the arc $(E(v)\cup \{x\},E(v)\cup \{x,y\})$ with  the arc $(E(v)\cup \{x\},E(G) \cup \{x\})$.
This results in a branching $B'$ such that $e \in \cup N^-_{B'}(E(G) \cup \{x\})$.
Moreover, $|U(B')|\le|U(B)|$ since removing the arc $(E(v)\cup \{x\},E(v)\cup \{x,y\})$ makes $x$  uncovered in $E(v)\cup \{x,y\}$, but adding the arc $(E(v)\cup \{x\},E(G) \cup \{x\})$ makes element $e$ covered in $E(G) \cup \{x\}$.
Therefore, repeating the above procedure will eventually result in an optimal branching with respect to which
$E(G) \cup \{x\}$ is not irreducible, as claimed.

Define $C= \{ v \in V(G) : (E (v) \cup \{x \}, E(G) \cup \{ x\}) \in  B\}$. The fact that every $e\in E(G)$ is covered in $E(G)\cup \{x\}$ implies that
$C$ is a vertex cover of $G$. Moreover, for every $v\in C$, element $x$ is the only uncovered element in vertex $E(v)\cup \{x,y\}$, and for every
$v\in V(G)\setminus C$, all elements in $E(v)\cup \{x,y\}$ are covered. This implies that the total number of uncovered pairs by $B$ equals
$8|V(G)|+|C|$, implying $|C| = \beta(M) - 8|V(G)|$, which proves the claimed inequality
$\tau(G)\le \beta(M) - 8|V(G)|$.

This completes the proof of the claim.
\end{proof}

We now complete the proof by showing that the above reduction is an $L$-reduction.
Since $G$ is cubic, every vertex in a vertex cover of $G$
covers exactly $3$ edges, hence $\tau(G)\geq \frac{|E(G)|}{3} = \frac{|V(G)|}{2}$.
This implies that $\beta(M) = \tau(G) + 8|V(G)|\le 17\tau(G)$, hence the first condition
in the definition of $L$-reducibility is satisfied with $a = 17$.
The second condition in the definition of $L$-reducibility
states that for every branching $B$ of $D_M$ we can
can compute in polynomial time a vertex cover $C$ of $G$
such that $|C|-\tau(G)\le b \cdot (|U(B)|-\beta(M))$ for some $b>0$.
We claim that this can be achieved with $b = 1$. Indeed, the second part of the proof of above claim
shows how one can transform in polynomial time any branching of $D_M$
into a vertex cover $C$ of $G$ such that $|C|\le |U(B)|-8|V(G)|$.
Therefore, $|C|-\tau(G)\le |U(B)|-8|V(G)|-\tau(G) = |U(B)|-\beta(M)$.
This shows that the vertex cover problem in cubic graphs is $L$-reducible to
the  {\sc MinimumUcoveringBranching} and completes the proof.
\end{proof}
\begin{proposition}\label{prop:apxHard2}
The MIB problem is \apx-hard, even for instances of height~$2$.
Consequently, the MDCRS problem is \apx-hard, even for instances of height~$2$.
\end{proposition}

\begin{proof}
We construct an $L$-reduction from the vertex cover problem in cubic graphs to the MIB problem.
Let $G$ be a cubic graph. Let $M$ be a binary matrix without duplicated columns such that its column hypergraph is isomorphic to $\mathcal{H}$, where $\mathcal{H} = (E, E\cup  \{E(x) : x \in V\})$. See Fig.~\ref{fig:construction-2} for an example construction, representing the containment digraph $D_M$ of the binary matrix derived from the complete graph $K_4$.

\begin{figure}[h!]
\begin{center}
\includegraphics[width=0.9\textwidth]{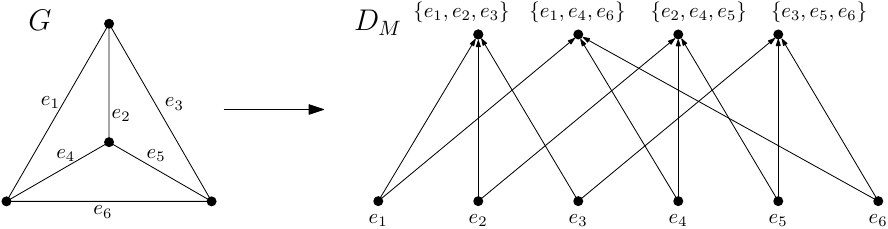}
\end{center}
\caption{An example construction of the $L$-reduction from the proof of Proposition~\ref{prop:apxHard2}.}
\label{fig:construction-2}
\end{figure}

To prove \apx-hardness, we will show that $\zeta(M)  = |E(G)| + \tau(G)$.
This will suffice: since every vertex in a vertex cover covers at most three edges, we have $\tau(G)\ge |E(G)|/3$, which
will imply that $\zeta(M)\le 4\tau(G)$. Similar arguments as those used at the end of the proof
of Proposition~\ref{prop:apxHard} can then be used to infer that the given reduction is an $L$-reduction, thus completing the proof of the theorem.

We split the proof of  $\zeta(M)  = |E(G)| + \tau(G)$ into two parts. First we show that $\zeta(M)  \le |E(G)| + \tau(G)$. Let $C$ be any minimum vertex cover of $G$. Define a set of arcs $B$ of $D_M$ as $B = \{(e, E(x)) : x \in e \wedge x \in V(G) \setminus C\}$. We first claim that $B$ is branching of $D_M$. Indeed, if this was not the case, then there would exist an edge $e \in E(G)$ and two distinct vertices $x,y \in V(G)$ such that $(e,E(x)), (e, E(y)) \in B$. This would imply that $e \in E(x)$ and $e \in E(y)$ and consequently $e = xy$. By definition of $B$, none of $x$ and $y$ is in $C$, contradicting the fact that $C$ is vertex cover.

Let $x\in V(G)$. We claim that $E(x) \in I(B)$ implies that $x \in C$. Suppose for a contradiction that $E(x) \in I(B)$ with $x \not \in C$. Since $x\in V(G)\setminus C$, the definition of $B$ implies that $(e,E(x))\in B$, for every $e\in E(x)$, in particular, every element of $E(x)$ is $B$-covered in $E(x)$. Hence $E(x)\not \in I(B)$, a contradiction. This shows that $|I(B) \cap \{E(x): x \in V(G)\}| \le |C|$. Together with $I(B)=(I(B) \cap E(G)) \cup (I(B) \cap \{E(x) : x  \in V(G)\})$ this implies that $|I(B)| \le |E(G)| + |C|$. It follows that $\zeta(M)\le |I(B)|\leq |E(G)|+|C| = |E(G)|+\tau(G)$, as claimed.

Next we show that $\zeta(M) \ge |E(G)| + \tau(G)$ by showing that $\tau(G) \le \zeta(M) - |E(G)|$.
Let $B$ be a branching of $D_M$ such that $|I(B)|= \zeta(M)$.
Define a set $C$ with $C = \{x \in V(G) : E(x) \in I(B)\}$. We claim that $C$ is a vertex cover of $G$.
Suppose that this does not hold, that is, that there exists $e\in E(G)$, such that $e=xy$ and $x,y \in V(G)\setminus C$. Since $x,y\not \in C$, it follows that $E(x), E(y) \not \in I(B)$. By construction, every element of $D_M$ of the form $E(z)$ is $B$-irreducible, unless $B$ contains all the three arcs leading to $E(z)$.
Consequently, $B$ contains all the three arcs leading to $E(x)$, and similarly for $E(y)$. In particular, we infer that
 $(e, E(x)), (e, E(y)) \in B$, contradicting the fact that $B$ is a branching in $D_M$.
Since $I(B)$ is the disjoint union of $I(B) \cap E(G)$ and  $ I(B) \cap \{E(x) : x  \in V(G)\}$ and   $E(G)\subseteq I(B)$ we have
$|I(B)|= |I(B) \cap E(G)| + |I(B) \cap \{E(x) : x  \in V(G)\}| = |E(G)| + |C|$, implying that $\tau(G)\le |C|=|I(B)|-|E(G)|=\zeta(M)-|E(G)|$.
\end{proof}

\begin{sloppypar}
\begin{theorem3}
The MUB and the MIB problems (and consequently the MCRS and the MDCRS problems) are \apx-hard, even for instances of height~$2$.
\end{theorem3}
\end{sloppypar}

\begin{proof}
The theorem combines the statements of Propositions~\ref{prop:apxHard} and~\ref{prop:apxHard2}.
\end{proof}

\subsection{$2$-Approximating $\eta$ and $\zeta$ via Laminar Set Families}\label{sec:2-approx}

\begin{sloppypar}
The result of Theorem~\ref{thm:apxHard} raises the question whether the four problems (MCRS, MDRCS, MUB, and MIB) admit constant factor approximations. In this section, we show that this is the case for the MDRCS and the MIB problems. This will be achieved by proving a lower and an upper bound for $\eta(M)$, which will together imply a simple \hbox{$2$-approximation} algorithm.

The lower bound is based on a connection between conflict-free matrices and laminar set families and an upper bound on the size of a laminar family in terms of the size of the ground set. Recall that a hypergraph $\mathcal{H}$ is said to be {\em laminar} if every two hyperedges $e_1,e_2 \in E(\mathcal{H})$ satisfy $e_1 \cap e_2 = \emptyset$, $e_1 \subseteq e_2$, or $e_2 \subseteq e_1$. Recall also that the column hypergraph $\mathcal{H}_M$ of a binary matrix $M$ is the hypergraph with vertex set $V(\mathcal{H}_M) = R_M$ and hyperedge set $E(\mathcal{H}_M) = \{\supp_M(c): c\in C_M \}$.
\end{sloppypar}

The following observation follows immediately from definitions.

\begin{observation}\label{observation:laminar}
A binary matrix $M$ is conflict-free if and only if its column hypergraph $\mathcal{H}_M$ is laminar.
\end{observation}

The following upper bound on the size of a laminar hypergraph is well known, see, e.g.,~\cite{schrijver}.

\begin{theorem}\label{thm:laminar}
Every laminar hypergraph $\mathcal{H}$ satisfies $|E(\mathcal{H})| \le 2|V(\mathcal{H})|$.
\end{theorem}

Observation~\ref{observation:laminar} and Theorem~\ref{thm:laminar} imply the following.

\begin{corollary}\label{corollary:lbound}
Every conflict-free binary matrix $M$ with $m$ rows satisfies $k \le 2m$, where $k$ is the number of distinct columns of $M$.
\end{corollary}

The claimed $2$-approximation will be based on three lemmas.

\begin{lemma}\label{lemma:numOfDistColumnsM'}
If $M'$ is a conflict-free row split of $M$, then the number of distinct columns of $M'$ is at least as large as
the number of distinct columns of $M$.
\end{lemma}

\begin{proof}
It suffices to prove that each two distinct columns of $M$ are still distinct after performing the row split. Let $c_i, c_j$ be two distinct columns of $M$ and $c_i', c_j'$ the corresponding columns of $M'$. Then, without loss of generality, there exists a row $r$ of $M$ such that, $M_{r,i} = 0$ and $M_{r,j} =1$. Let $R(r)$ be the set of split rows of $r$ with respect to $M'$. Then for every $r' \in R(r)$ it holds $M'_{r', i} = 0$. Since the rows in $R(r)$ split $r$, there exists some $r'' \in R(r)$ with $M'_{r'', j}=1$.  This gives us $M'_{r'', i}=0$ and $M'_{r'',j} = 1$, showing that columns $c_i'$ and $c_j'$ are distinct.
\end{proof}

The following lemma shows that the value of $\eta$ is invariant under deleting one of a pair of identical columns.

\begin{lemma}\label{lemma:reduction}
For every binary matrix $M$ it holds that
$$\gamma(M)=\gamma(\rd(M))\,, \qquad  \eta(M)=\eta(\rd(M))\,,$$
\vspace{-4mm}
$$\beta(M)=\beta(\rd(M))\,,  \qquad  \zeta(M)=\zeta(\rd(M))\,.$$
\end{lemma}

\begin{proof}
Since $\rd(M)$ is submatrix of $M$, it follows that $\gamma(\rd(M))\leq \gamma(M)$ and, similarly, that $\eta(\rd(M))\leq \eta(M)$.
Conversely, since any conflict-free row split of $\rd(M)$ can be transformed to a conflict-free row split of $M$ with the same number of rows (by duplicating some columns) it follows that $\gamma(M)\leq \gamma(\rd(M))$ and $\eta(M)\leq \eta(\rd(M))$.
We have shown that $\gamma(M)=\gamma(\rd(M))$ and $\eta(M)=\eta(\rd(M))$.
Moreover, since the containment digraphs $D_M$ and $D_{\rd(M)}$ are the same, we infer that $\beta(M)=\beta(\rd(M))$ and $\zeta(M)=\zeta(\rd(M))$.
\end{proof}

Corollary~\ref{corollary:lbound} and Lemmas~\ref{lemma:numOfDistColumnsM'} and~\ref{lemma:reduction} together with a simple row splitting strategy imply the following.

\begin{lemma}\label{lem:bounds-eta}
For every binary matrix $M$, we have $k/2 \le \eta(M) \le k$, where $k$ is the number of distinct columns of $M$.
\end{lemma}

\begin{sloppypar}
\begin{proof}
Let $M \in \{0,1\}^{m \times n}$.
First, we prove that  $k/2 \le \eta(M)$ or, equivalently, that $k \le 2 \eta(M)$.
Let $M' \in \{0,1\}^{m' \times n}$ be a row split of $M$ with exactly $\eta(M)$ distinct rows.
Let $k'$ be the number of distinct columns of $M'$.
Let $N\in \{0, 1\}^{\eta(M) \times n}$ be a new matrix obtained from $M'$ by taking one row from each set of identical rows.
It is not difficult to see that $N$  is conflict-free, with exactly $k'$ distinct columns.
Further on, by Corollary~\ref{corollary:lbound} it holds $k'\le 2 \eta(M)$ and hence by Lemma~\ref{lemma:numOfDistColumnsM'} it holds $k \le k' \le 2\eta(M)$, as claimed.

It remains to show $\eta(M) \le k$. By Lemma~\ref{lemma:reduction} it suffices to show that $\eta(\rd(M)) \le k$. Let $M'$ be the row split of $\rd(M)$ obtained by splitting each row $r$ with $t$ ones into $t$ rows, each with exactly one non-zero entry.
By construction, $M'$ has exactly $k$ columns and therefore at most $k$ distinct rows. It follows that $\eta(\rd(M))\le k$, as desired.
\end{proof} \end{sloppypar}

Now we have everything ready to state and prove the announced approximation result.

\begin{theorem}\label{thm:2-approx}
There is a $2$-approximation algorithm for the MDCRS (and consequently for the MIB)
problem running in time ${\mathcal{O}}(mnk)$ on a given matrix $M\in \{0,1\}^{m\times n}$
where $k$ is the number of distinct columns of $M$.
\end{theorem}

\begin{proof}
Let $M$ be a binary matrix with $m$ rows and $n$ columns, exactly $k$ of which are distinct. The proof of Lemma~\ref{lem:bounds-eta} is constructive and leads to the following algorithm to compute a row split of $M$ with at most $k$ distinct rows:
\begin{enumerate}
\item Compute $\rd(M)$. (This can be done in time ${\mathcal{O}}(mn)$ using radix sort.)
\item Compute a row split $M'$ of $\rd(M)$ obtained by splitting each row $r$ with $t$ ones into $t$ rows, each with exactly one non-zero entry. (This can be done in time ${\mathcal{O}}(mk^2)$.)
\item Transform $M'$ into a row split of $M$ by an appropriate duplication of some columns.  (This can be done in time ${\mathcal{O}}(kmn)$, since $M'$ has at most $km$ rows and the constructed matrix will have exactly $n$ columns.)
\end{enumerate}
Clearly, the algorithm produces a row split of $M$ with at most $k$ distinct rows. Since $\eta(M)\ge k/2$, it follows that this is a $2$-approximation.
Moreover, using the fact that $k\le n$, we infer that the total time complexity of the algorithm is
${\mathcal{O}}(mn + mk^2 + mnk) = {\mathcal{O}}(mnk)$, as stated.
\end{proof}

Note that Theorems~\ref{thm:apxHard} and~\ref{thm:2-approx} imply that the MDCRS and the MIB problems are \apx-complete.

\subsection{Two Approximation Algorithms for Computing $\gamma$ and $\beta$}\label{subsect:height and width}

While the question of whether the MCRS (and consequently the MUB) problem admits a constant factor approximation algorithm on general instances remains open, we give in this section two partial results in this direction. We show that
the two problems admit constant factor approximation algorithms on instances of bounded height or width.

Roughly speaking, the following theorem shows that for instances of bounded height {\it any} algorithm
for the MCRS problem based on branchings is a constant factor approximation algorithm.

\begin{theorem}\label{thm:height}
Let $M$ be a binary matrix and let $B$ be an arbitrary branching of $D_M$.
Then, the number of rows in the $B$-split of $M$ is at most $h(M)\gamma(M)$.
\end{theorem}

\begin{proof}
Let $h = h(M)$ and let $B_{opt}$ be a branching of $D_M$ with $|U(B_{opt})|=\beta(M)$.
Recall that the number of rows in the $B$-split of $M$ is $|U(B)|$.  Since by  Theorem~\ref{thm:beta=gamma} $\beta(M)=\gamma(M)$, it suffices to prove that $|U(B)|\leq h\beta(M)$.
For $(r,v)\in U(B_{opt})$, define the set $\Omega(r,v)$ with
$$
\Omega(r,v)=\{(r,v')\mid v \in V, v'\in B_{opt}^{+}(v)\}.
$$
We claim that
\begin{equation}\label{eq:omega(r,v)}
U(B)\subseteq\cup_{(r,v)\in U(B_{opt})}\Omega(r,v).
\end{equation}
(In fact, since $U(B)=U(\emptyset)=\{(r,v):r\in v \in V\}$, equality holds in \eqref{eq:omega(r,v)}, but we will not need it in the proof.)
Let $(r,v')\in U(B)$.
We will show that there exists some $(r,v)\in U(B_{opt})$ such that $(r,v')\in \Omega(r,v)$.
If $(r,v')\in U(B_{opt})$, then $(r,v')\in \Omega(r,v')$, since $v'\in B_{opt}^+(v')$.
If $(r,v')\not \in U(B_{opt})$, it follows that $r$ is covered in $v'$ with respect to $B_{opt}$, and therefore, there exists some $v''$ such that $(v'',v')\in B_{opt}$, and $r\in v''$. If $(r,v'')\in U(B_{opt})$, then it is clear that $(r,v')\in \Omega(r,v'')$. If $(r,v'')\not \in U(B_{opt})$, then we repeat the described procedure, which has to terminate after finitely many steps. Therefore, there exists some $(r,v)\in U(B_{opt})$ such that $(r,v')\in \Omega(r,v)$, as claimed. This establishes inclusion \eqref{eq:omega(r,v)}.

Since the height of $D_M$ is $h$, it follows that the height of  $B_{opt}$ is at most $h$. Moreover, since $B_{opt}$ is a branching, it follows that $|\Omega(r,v)|\le h$, for every $(r,v)\in U(B_{opt})$. Combining this with \eqref{eq:omega(r,v)}, we have
$$
|U(B)|\leq \sum_{(r,v)\in U(B_{opt})} |\Omega(r,v)|\leq h|U(B_{opt})|=h\beta(M).
$$
Since  $\beta(M)=\gamma(M)$, it follows that $|U(B)|\leq h\gamma(M)$. This completes the proof.
\end{proof}

\begin{sloppypar}
\begin{remark}
Since in Theorem~\ref{thm:height}, there is no restriction on the branching $B$,
an $h(M)$-approximation to $\gamma(M)$ can be obtained simply by taking $B = \emptyset$ and returning the resulting row split.
\end{remark}
\end{sloppypar}

\begin{remark}\label{remark:tight}
The following example shows that for every $h>1$ and every $\epsilon>0$, the algorithm for the MCRS problem given by Theorem~\ref{thm:linear-branchings} is not an $(h-\epsilon)$-approximation when restricted to instances of height $h$.
\end{remark}

\begin{example}\label{ex:d-ary tree}
Fix a positive integer $h\geq 2$. For all $d\ge 2$, let $M_d$ be a binary matrix with $d^{h-1}$ rows (indexed by $1,\ldots,d^{h-1}$) and $1+d+d^2+\ldots+d^{h-1}$ columns. Entries of $M_d$ are defined so that the supports of columns of $M_d$ are given by $A_j^i=\{(j-1)d^{i-1}+k: 1\leq k\leq d^{i-1}\}$ for all $i\in \{1,\ldots,h\}$ and all $j\in \{1,\ldots,d^{h-i-1}\}$ (see Figs.~\ref{fig:tree} and~\ref{fig:tree2} for an example). The height of $M_d$ is $h$.
Matrix $M_d$ is conflict-free, therefore $\gamma(M_d)=\beta(M_d)=d^{h-1}$. It can be seen that $\beta_{\ell}(M_d)=h\cdot d^{h-1}-(h-1)\cdot d^{h-2}$. Therefore,
$\beta_{\ell}(M_d)/\beta(M_d)=h-\frac{h-1}{d}$.
It follows that for every $\epsilon >0$, we have $\beta_{\ell}(M_d)/\beta(M_d)>h-\epsilon$ for all large enough $d$.
\end{example}

\begin{figure}[h!]
\begin{center}
\includegraphics[width=0.9\textwidth]{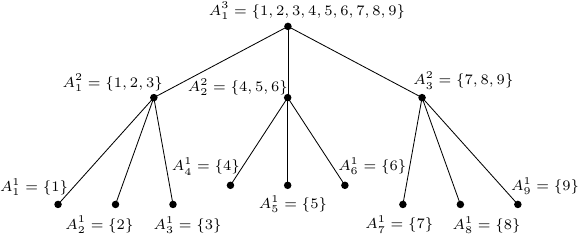}
\end{center}
\caption{The underlying tree of a branching $B$ of $D_{M_d}$ where $M_d$ is defined as in Example~\ref{ex:d-ary tree} for $d=3$ and $h=3$.}
\label{fig:tree}
\end{figure}

\begin{figure}[h!]
\begin{center}
\includegraphics[width=0.9\textwidth]{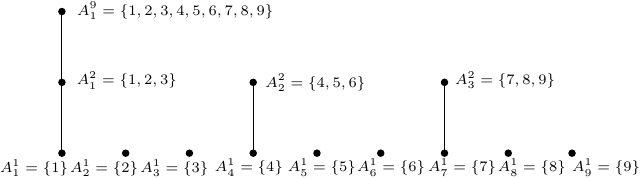}
\end{center}
\caption{The underlying forest of an optimal linear branching $B'$ given by Theorem~\ref{thm:linear-branchings}.
Note that, if $B$ denotes the branching from Fig.~\ref{fig:tree}, then $|U(B)|=\gamma(M_d)=9$ and $|U(B')|=21$.}
\label{fig:tree2}
\end{figure}

For instances of bounded width, a constant factor approximation can be obtained by considering any $B$-split resulting from a linear branching $B$ of $D_M$ consisting of ${\it wdt}(M)$ paths. Note that such a branching can be computed in polynomial time using Dilworth's theorem (Theorem~\ref{thm:Dilworth}).

\begin{theorem}\label{thm:width}
Any algorithm that, given a binary matrix $M$, computes a linear branching $B$ of $D_M$ consisting of ${\it wdt}(M)$ paths and returns
the corresponding $B$-split of $M$ is a ${\it wdt}(M)$-approximation algorithm for the MCRS problem.
\end{theorem}

\begin{proof}
Let $M$ be a binary matrix and let $w={\it wdt}(M)$. Let $P=\{C_1,\ldots,C_{w}\}$ be a chain partition of $D_M$ (the existence of such a partition
is guaranteed by Dilworth's theorem) and let $B$ be the linear branching of $D_M$ corresponding to $P$.
We will prove that $|U(B)|\leq w |R_M|$, where $R_M$ denotes the set of rows of $M$.
We claim that the number of elements in $U(B)$ with fixed first coordinate is at most $w$. For a row $r$ of $M$, let $N(r)=\{v\in V(D_M):(r,v)\in U(B)\}$.
We claim that $|N(r)\cap C_i|\leq 1$, for every $i\in \{1,\ldots,w\}$.
Suppose that $v_1\neq v_2$ and $v_1,v_2\in N(r)\cap C_i$ for some $i\in \{1,\ldots,w\}$. Since $C_i$ is a chain, we may assume without loss of generality that $v_1\subset v_2$. Moreover, since  $v_1,v_2$ are both in $C_i$ it follows that
there exists a path in $B$ from $v_1$ to $v_2$.
Since $v_1\in N(r)$, it follows that $r\in v_1$, and since there exists a path in $B$ from $v_1$ to $v_2$, it follows that $r$ is covered in $v_2$ with respect to $B$. This contradicts the assumption that $v_2\in N(r)$. The obtained contradiction
shows that $|N(r)\cap C_i|\leq 1$, as claimed.
Since $|N(r)\cap C_i|\leq 1$, and $P$ is a chain partition of $D_M$, it follows that $|N(r)|= \sum_{i=1}^w|N(r)\cap C_i|\leq w$.
It is now easy to see that $|U(B)|=\sum_{r\in R_M}|N(r)|\leq w|R_M|$.

Since matrix $M$ is assumed to have no row whose all entries are $0$, every row split of $M$ contains at least $|R_M|$ rows, that is,
$|R_M|\le \gamma(M)$. It follows that $|U(B)|\leq w\gamma(M)$ and since the $B$-split of $M$ has exactly $|U(B)|$ rows (by Lemma~\ref{lemma:branching to split}), the claimed approximation ratio follows.
\end{proof}

\section{Conclusion}\label{sec:conclusion}

In this paper, we revisited the minimum conflict-free row split problem and a variant of it. We formulated the two problems as optimization problems on branchings in a derived directed acyclic graph and, building on these formulations, obtained several new algorithmic and complexity insights about the two problems, including \apx-hardness results and approximation algorithms. Moreover, we proved a min-max result on digraphs strengthening the classical Dilworth's theorem and leading to a new heuristic for the MCRS problem. In Figure~\ref{fig:problems} we summarize the relations between several problems discussed in this paper, along with known complexity results and some applications.
The relations are described informally; for instance, we say that
problem $P_1$ \emph{reduces to} problem $P_2$ if a polynomial-time algorithm
for problem $P_2$ can be used to develop a polynomial-time algorithm for problem $P_1$.

\begin{figure}[h!]
\begin{center}
\includegraphics[width=\textwidth]{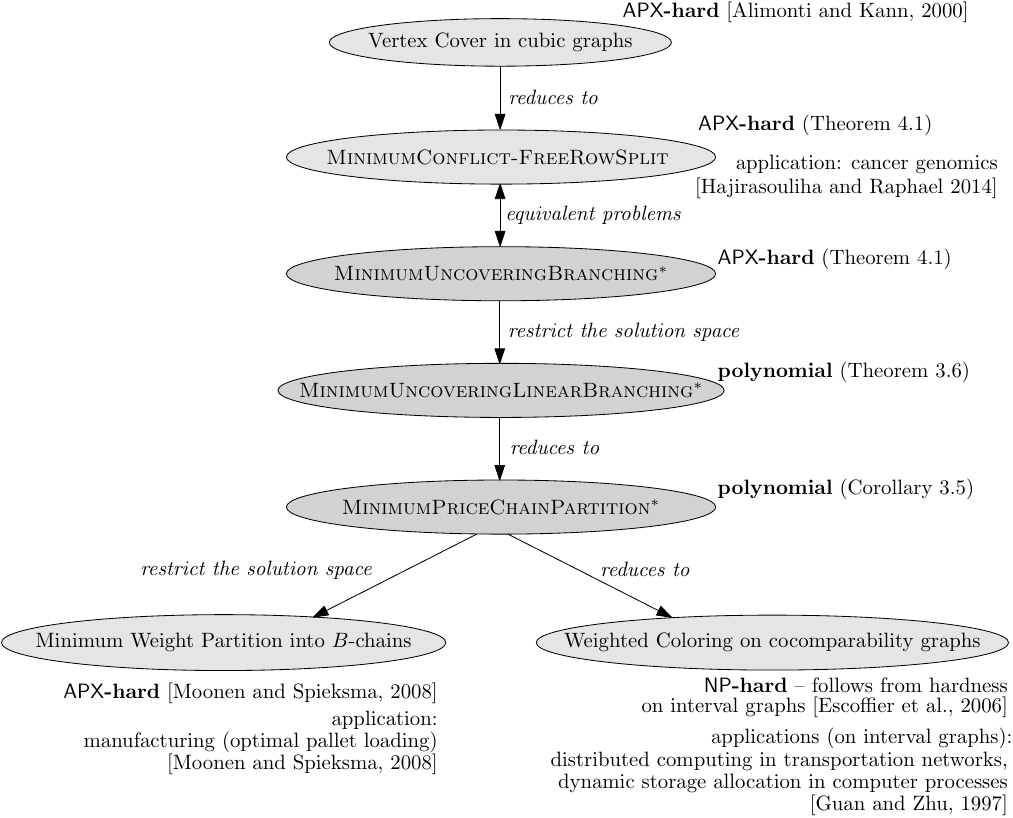}
\end{center}
\caption{Summary of relations between some problems discussed in this paper,
along with known complexity results and some applications. Problems marked with an asterisk
($^\ast$) are introduced in this paper.}
\label{fig:problems}
\end{figure}

The main problem left open by our work is the determination of the exact (in)approximability status of the MCRS problem.
In particular, does the problem admit a constant factor approximation?
Other possibilities for related future research include:
i) the study of the approximability properties of the closely related Minimum-Split-Row problem~\cite{wabi14}
(our preliminary investigations show that the problem, while being \apx-hard, admits a $(2h(M)-1)$-approximation);
ii) a parameterized complexity study of the considered problems (along with identification of meaningful parameterizations), and iii) a study of extensions of the model that could be relevant for the biological application, such as the case when the input binary matrix may contain errors or has partially missing data. Finally, it would be interesting to find further applications of the
polynomially solvable {\sc MinimumPriceChainPartition} problem,  as well as of the two branching problems,  {\sc MinimumUncoveringBranching} and {\sc MinimumIrreducingBranching}, introduced in this paper.

\subsection*{Acknowledgments}

\begin{sloppypar}
The authors are grateful to the two anonymous reviewers for helpful remarks. This work is supported in part by the Slovenian Research Agency under research programs I0-0035, P1-0285, and research projects N1-0032, N1-0038, N1-0062, J1-6720, J1-7051, and by the Academy of Finland
under Grant No.~274977.
\end{sloppypar}

\end{document}